\documentclass{article}

\PassOptionsToPackage{numbers,square}{natbib}
\usepackage[final]{neurips_2021}
\usepackage[utf8]{inputenc} 
\usepackage[T1]{fontenc}    
\usepackage{hyperref, soul}       
\usepackage{url}            
\usepackage{booktabs}       
\usepackage{amsfonts}       
\usepackage{nicefrac}       
\usepackage{microtype}      
\usepackage{tabularx,enumerate}
\usepackage[dvipsnames]{xcolor}
\usepackage{amsmath,amssymb,dsfont,subcaption,algcap}
\usepackage{lipsum,bm}
\usepackage{longtable}
\usepackage{cleveref}

\usepackage{array}
\newcolumntype{H}{>{\setbox0=\hbox\bgroup}c<{\egroup}@{}}
\usepackage{chngcntr,multirow}
\usepackage{caption}
\usepackage{subcaption}
\setcitestyle{citesep = {;}}

\usepackage{mathtools, graphicx,color,dsfont,amsthm,bbm}
\DeclareMathOperator*{\argmin}{arg\,min}
\DeclarePairedDelimiter\floor{\lfloor}{\rfloor}
\theoremstyle{plain}
\newtheorem{lemma}{Lemma}
\newtheorem{thm}{Theorem}
\newtheorem{asp}{Assumption}
\newtheorem{prop}{Proposition}
\newtheorem{cor}{Corollary}

\newcommand{\mP}{\mathbb{P}}

\usepackage{algorithm}
\usepackage{algpseudocode}
\algnewcommand\INPUT{\item[\algorithmicinput]}
\algnewcommand\OUTPUT{\item[\algorithmicoutput]}
\graphicspath{ {Figures/} }

\theoremstyle{definition}
\newtheorem{defn}{Definition}

\everymath{\displaystyle}

\title{Adversarially Robust Change Point Detection}

\author{%
  Mengchu Li \\
  Department of Statistics\\
  University of Warwick\\
  \texttt{mengchu.li@warwick.ac.uk} \\
  \And
  Yi Yu \\
  Department of Statistics \\
  University of Warwick \\
  \texttt{yi.yu.2@warwick.ac.uk}
}

\begin{document}

\maketitle

\begin{abstract}
Change point detection is becoming increasingly popular in many application areas.  On one hand, most of the theoretically-justified methods are investigated in an ideal setting without model violations, or merely robust against identical heavy-tailed noise distribution across time and/or against isolate outliers; on the other hand, we are aware that there have been exponentially growing attacks from adversaries, who may pose systematic contamination on data to purposely create spurious change points or disguise true change points.   In light of the timely need of a change point detection method that is robust against adversaries, we start with, arguably, the simplest univariate mean change point detection problem.  The adversarial attacks are formulated through the Huber $\varepsilon$-contamination framework, which in particular allows the contamination distributions to be different at each time point. In this paper, we demonstrate a phase transition phenomenon in change point detection. This detection boundary is a function of the contamination proportion~$\varepsilon$ and is the first time shown in the literature.  In addition, we derive the minimax-rate optimal localisation error rate, quantifying the cost of accuracy in terms of the contamination proportion.  We propose a computationally-feasible method, matching the minimax lower bound under certain conditions, saving for logarithmic factors.  Extensive numerical experiments are conducted with comparisons to existing robust change point detection methods.
\end{abstract}

\section{Introduction}\label{sec-intro}

Change point detection is attracting tremendous attention due to the demand from various application areas, including bioinformatics \citep[e.g.][]{olshen2004circular,futschik2014multiscale}, climatology \citep[e.g.][]{reeves2007review,gallagher2013changepoint} and finance \citep[e.g.][]{pepelyshev2015real,pastor2001equity}, among many others.  In the last few decades, a vast body of methods and theory on change point analysis have been studied based on different data types \citep[e.g.][]{page1954continuous,aue2009break,killick2012optimal,wang2018high,wang2018optimal,wang2021optimal}.  The majority of the methods are studied in a model-specific way, in the sense that the noise distributions are sub-Gaussian/sub-Exponential and the between change points data are independent and identically distributed.  The few robust change point detection results \citep[e.g.][]{huvskova1991recursive,fisch2018linear,fearnhead2019} are designed against isolated outliers and/or heavy-tailed noise. 

In recent years, we are aware of the risk of adversary attacks in emerging application areas, ranging from image classification \cite{kurakin2017adversarial}, object detection \cite{song2018physical}, to natural language processing \cite{jia2017adversarial} and beyond. In this paper, we consider a change point analysis setting possibly attacked by adversaries. As a teaser, we consider a concrete climate data set from \cite{airquality} containing the daily average PM2.5 index data in Beijing from 15-Apr-2017 to 15-Feb-2021.  The original and adversarially contaminated  data (with a spurious change point created on 17-Jan-2020) are shown in Figure \ref{teaser}, where the orange points in the right panel denote the adversarially chosen contamination points.  As we will see more details in Sections \ref{realdata} and \ref{realdatasupp}, applying the \textsc{Biweight}(3) method \citep{fearnhead2019} to the original data set leads to two change point estimators, while only the spurious change point is detected in the presence of contamination. In contrast, our proposed method {a\textsc{arc}} is less affected by the adversarial attack and detects the same two change points with and without the presence of the contamination.

\begin{figure}[hbt]
    \centering
    \includegraphics[width=1\textwidth,height = 1.6in]{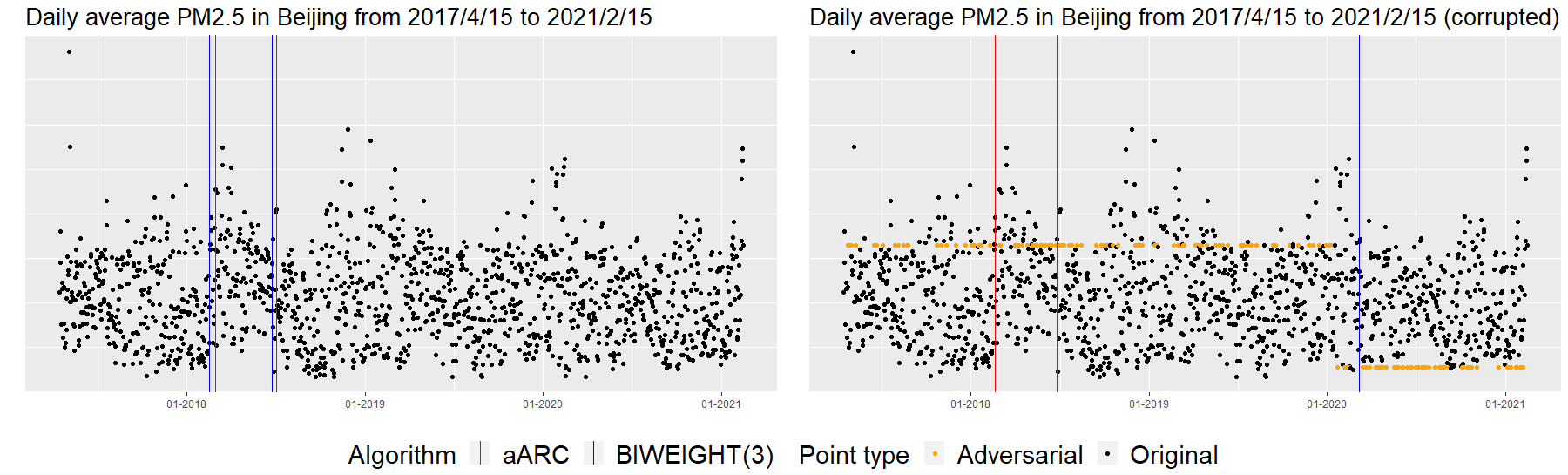}
\caption{Adversarial attacks on a PM2.5 index data set.  Red and blue lines indicate the change points estimators of  {a\textsc{arc}} and \textsc{Biweight}(3); lines in the left panel are jittered to be visible.} 
\label{teaser}
\end{figure}

From this example, we see that the presence of adversarial contamination can significantly affect the performance of even the {state-of-the-art} robust change point detection algorithm and lead to detecting spurious change points, while missing out the ones that can be discovered without contamination.  
In general, three effects of the adversarial attacks are of interest: \textbf{(i)} hiding true change points; \textbf{(ii)} creating spurious change points; and \textbf{(iii)} increasing the localisation error rate.

To armour the change point detection procedure against potential adversarial attacks, in this paper, we establish and investigate a univariate mean change point detection framework, under a dynamic extension of the Huber $\varepsilon$-contamination model \eqref{huber0}. Our contributions in this paper are threefold.
 
$\bullet$  To the best of our knowledge, this is the first paper formalising the change point detection framework with the dynamic Huber $\varepsilon$-contamination model (\ref{huber}).  To be specific, the contamination distributions are allowed to be distinct at every time point.   Most if not all of the robust change point detection papers, despite that they have shown their methods are robust against outliers and/or heavy-tail noise, the theoretical framework they study are still within the i.i.d.~territory.  This has handicapped the existing work to study the adversary attacks, where adversaries may design specific contamination based on their knowledge of the underlying models.  
     
$\bullet$  In Section \ref{lowerboundsection}, we propose a signal-to-noise ratio quantity $\kappa/\sigma$ and show that in the regime $ \kappa/\sigma \lesssim \sqrt{\max\left\{\varepsilon, \log(n)/L \right\}},$ no algorithm is guaranteed to estimate the change points consistently in the sense of \eqref{consistency}; and we show in Section \ref{section3} that in the regime $ \kappa/\sigma \gtrsim \sqrt{\max\left\{\varepsilon, \log(n)/L \right\}}$ and $L \gtrsim \log(n)$, our proposed algorithm can localise change points consistently with properly chosen tuning parameters. The localisation error rate can also be nearly minimax optimal, off by a logarithmic factor, for a certain range of model parameters. The detection boundary matches that in the standard change point literature (i.e.~$\varepsilon = 0$) and that in the robust statistics literature (i.e.~$\kappa = 0$). Compared to the standard literature in the univariate mean change point detection~\citep{verzelen2020optimal,wang2020}, our results can quantify the cost of the contamination in terms of $\varepsilon$, and shed light on both the cost of robustness and designs of adversarial attacks.
    
$\bullet$  The Adversarially Robust Change point detection algorithm (\textsc{arc}) that we propose in Algorithm~\ref{alg-main} is a combination of \cite{prasad2019unified} and a simple scanning idea.  \citet{prasad2019unified} showed that their robust univariate mean estimator (\textsc{rume}) can provide optimal estimation without the presence of change points.  The scanning idea has been widely used with numerous variants, but none of which is studied in an adversarial setting.  In our paper, we exploit the potential of these two areas and investigate both the theoretical and numerical performances of \textsc{arc}.  In addition, we also propose a variant of \textsc{arc}, namely automatic \textsc{arc} (a\textsc{arc}), which adapts to the contamination proportion $\varepsilon$.

\subsection{Related literature}

Without the concern of robustness, the theoretical framework of univariate mean change point analysis problem is well established \citep[e.g.][]{wang2020, verzelen2020optimal}, with a host of algorithms available to practitioners, including penalised least square methods \citep[e.g.][]{killick2012optimal} and \textsc{cusum}-based methods \citep[e.g.][]{fryzlewicz2014wild}, among many others. 

When the robustness comes into play, a line of attack has been deployed recently.  \citet{fearnhead2019} considered a penalised $M$-estimation procedure, which is designed particularly against outliers with large variances.  \citet{yu2019robust} studied a testing problem, utilising a $U$-statistic-type test statistics and showing that it is robust against i.i.d.~heavy-tailed noise distributions.  \citet{huvskova2013robust} and \citet{huvskova1991recursive} adapted Huber's theory on robust statistics to study a regression change point detection problem.  There has also been work on the online version of the robust change point detection problem \citep[e.g.][]{knoblauch2018doubly}.

In the robust statistics literature, without the presence of change points, the heavy-tailed model and the Huber $\varepsilon$-contamination model are the two in the spotlight.  A fundamental problem therein is to estimate the mean of the underlying distribution or the decontaminated distribution accurately.  Efficient and optimal algorithms have been developed for both models separately \citep[e.g.][]{lugosi2019sub,diakonikolas2017being}.  Some work tackles the two models simultaneously.  For instance, \citet{prasad2019unified} exploited the connection between these two models and developed a computationally-efficient univariate mean estimator that is optimal under both models.  \citet{hopkins2020robust} and \citet{diakonikolas2020outlier} considered a filter-type algorithm, which was developed under the high-dimensional contamination model, and showed that it achieves optimal error guarantees under the heavy-tailed model as well.  It is worth noting that the theoretical results derived in this paper is not a straightforward adaptation of the existing literature due to the presence of potentially multiple change points. 

\subsection{Problem setup}

We kick off the formalisation of the problem with the Huber $\varepsilon$-contamination model \cite{huber1992robust}
\begin{equation}\label{huber0}
    F_\varepsilon = (1-\varepsilon)F + \varepsilon H,
\end{equation}
where $F$ and $H$ are the distributions of interest and arbitrary contamination, respectively, and $\varepsilon$ measures the strength of contamination. This model is widely used in the robust statistics literature, but usually if not always, it is assumed that $F$ is \textit{sub-Gaussian} or has some form of symmetry, and the data are \textit{i.i.d.}~from $F_\varepsilon$ \citep[e.g.][]{hampel2011robust,prasad2018robust,prasad2019unified}. In our model assumption below, we relax both restrictions by considering $F$ to be potentially heavy-tailed and allowing $H$ to vary across time, which we refer to as the \emph{dynamic Huber $\varepsilon$-contamination model}, and which prompts it to model the three types of adversarial attack we mentioned previously, as well as a wide range of less adversarial attacks studied already in the literature.

\begin{asp}\label{modelasmmp}
Let $\{Y_i\}_{i = 1}^n \in \mathbb{R}$ be a sequence of random variables with distributions
    \begin{equation}\label{huber}
        (1-\varepsilon_i)F_i+\varepsilon_i H_i, \quad i \in \{1, \ldots, n\},
    \end{equation}
    where $F_i$'s are distributions with means $f_i$'s and variances upper bounded by $\sigma^2<  \infty$, $H_i$'s are the distributions of arbitrary contamination and $\varepsilon_i$'s are the proportions of contamination upper bounded by $\varepsilon \in (0, 1/2)$. Let $\{\eta_k\}_{k = 0}^{K+1} \subset [0, n]$ be a strictly increasing integer sequence with $\eta_0 = 0$, $\eta_{K+1} = n$, satisfying that $f_{t+1} \neq f_{t}$, if and only if~ $t \in \{\eta_k\}_{k = 1}^K$. Further assume that $F_{\eta_{k}+1} = \dotsc = F_{\eta_{k+1}}$ for $k = 0,\dotsc,K$. Let the minimal spacing $L$ and jump size $\kappa$ be $L = \min_{k = 0}^{K} \{\eta_{k+1} - \eta_{k}\}$ and $\kappa = \min_{k=1}^K \kappa_{k} = \min_{k=1}^K |f_{\eta_k+1}-f_{\eta_k}|$. 
\end{asp}

With Assumption~\ref{modelasmmp}, our goal can be formalised as detecting any change on $f_i$'s, provided that $\kappa >0$, in the presence of adversarial noise $H_i$'s. In some applications, adversaries may have access to the generating process \eqref{huber}, and in particular have control over $H_i$'s.  As a consequence, adversaries may be able to design attacks creating spurious change points or cancelling out the change point patterns in $F_i$'s.  We are the first attempt in discussing the robust change point detection, allowing for such form of structural attacks.

We aim to obtain consistent change point estimators $\{\hat{\eta}_k\}_{k=1}^{\hat{K}}$ such that with high probability it holds that 
\begin{equation}\label{consistency}
    \widehat{K} = K \quad  \text{and} \quad n^{-1}d_\mathrm{H}\left(\{\eta_i\}_{i = 1}^K, \, \{\widehat{\eta}_i\}_{i = 1}^{\widehat{K}}\right) \to 0,
\end{equation}
where $d_\mathrm{H}(\cdot,\cdot)$ is the two-sided Hausdorff distance (see Definition~\ref{hausdorff} in the supplementary material).

In the case when $\kappa = 0$, i.e.~there is no change point on the signal $f_1,\dotsc,f_n$, regardless of the situations of $H_i$'s, we would like to have 
\[
\mP(\widehat{K} = 0) \rightarrow 1.
\]

To highlight that $F_i$'s are the distributions of interest, in Assumption~\ref{modelasmmp}, we restrict the power of an adversary by assuming $\varepsilon < 1/2$.  It is apparent that when $\varepsilon \geq 1/2$, detecting the change points in~$F_i$'s may be impossible regardless of any other model parameters.

\section{Phase Transition and Minimax Lower bounds}\label{lowerboundsection}

Without contamination, it is well established that, if $\kappa \sqrt{L}\sigma^{-1} \lesssim \sqrt{\log(n)}$, then in the minimax sense, no algorithm is guaranteed to produce consistent estimators; if $\kappa \sqrt{L}\sigma^{-1} \gtrsim \sqrt{\log(n)}$, then the optimal localisation rate is of order $\sigma^2\kappa^{-2}$  \citep[e.g.][]{wang2020, verzelen2020optimal}.  In this section, we are to present the counterparts of such results in the model described in Assumption~\ref{modelasmmp}.  

In detail, Lemma~\ref{lemma1} shows that if 
    \begin{equation}\label{stn}
        \kappa/\sigma \lesssim \sqrt{\max\left\{\varepsilon, \log(n)/L \right\}},
    \end{equation}
    then no consistent estimator exists.  In particular, by considering the two regimes inherited in \eqref{stn}, we can identify two sources of difficulties in detecting change points in an adversarial setting.
    
\textbf{Small $\varepsilon$ regime}. When $\varepsilon \lesssim \log(n)/L$, condition \eqref{stn} is reduced to $\kappa \sigma^{-1} \lesssim \sqrt{\log(n)/L}$, which is essentially the boundary without the presence of contamination (cf.~Section~2 in \cite{wang2020}).  An interesting interpretation of the result is that it quantifies how much contamination will affect the fundamental difficulty of the problem and this threshold is of order $\log(n)/L$. Similar phenomena have been observed in the robust mean estimation \citep[e.g.][]{Chen_2018} and testing \citep[e.g.][]{chen2016general} problems.
        
\textbf{Large $\varepsilon$ regime}. When $\varepsilon \gtrsim \log(n)/L$, condition \eqref{stn} is reduced to $\kappa/\sigma \lesssim \sqrt{\varepsilon}$. The term $\sigma\sqrt{\varepsilon}$ is essentially the minimal asymptotic bias that any estimator of the means $f_i$'s must suffer under the Huber $\varepsilon$-contamination model~\eqref{huber0} with finite variance $\sigma^2$ \cite{Lai_2016,prasad2019unified}. Here, we assert that if the jump size $\kappa$ is no larger than the asymptotic bias $\sigma\sqrt{\varepsilon}$, then no consistent estimator exists. Referring to the three attack strategies we mentioned in Section~\ref{sec-intro}, this is the situation where the signal of a change point can be completely hidden by contamination. 
    
\begin{lemma}\label{lemma1}  
Let $\{Y_i\}_{i = 1}^n$ satisfy Assumption~\ref{modelasmmp}.  Suppose that $\varepsilon_i = \varepsilon$, $i = 1, \dotsc, n$. Let $P^n_{\kappa,L,\sigma,\varepsilon}$ denote the corresponding joint distribution. Consider the class of distributions 
    \[
        \mathcal{P} = \left\{P^n_{\kappa,L,\sigma,\varepsilon}:\,   \kappa^2L\sigma^{-2} <  \max\left\{8\varepsilon/(1-2\varepsilon), \, \log(n), \, 4 \varepsilon L \right\}, \, L \leq \floor*{n/4} \right\}.
    \]
    For all $n \in \mathbb{N}$ large enough, it holds that $\inf_{\hat{\eta}} \sup_{P\in \mathcal{P}} \mathbb{E}_P\{d_\mathrm{H} (\widehat{\eta}, \, \eta(P))\} \geq n/8$, where the infimum is over all possible measurable functions of the data and $\eta(P)$ is the set of true change points of $P \in \mathcal{P}$.
\end{lemma}
   
Note that in many interesting scenarios, $\log(n) \geq 8\varepsilon/(1-2\varepsilon)$, which is equivalent to $ \varepsilon < \log(n)/\{2\log(n)+8\}$.  Therefore, for simplicity, in the following, we only consider the condition $\kappa/\sigma \lesssim \max\{\sqrt{\varepsilon}$, $\sqrt{\log(n)/L}\}$, under which, no algorithm is ensured to output a consistent estimator.  As we will show later in Theorem~\ref{scp2}, \textsc{arc} can produce consistent estimation under nearly optimal conditions.
   
Our second task is to demonstrate the optimal localisation error, under a higher signal-to-noise ratio condition.  To be specific, we consider
    \begin{equation}\label{lowercon}
        \min\left\{\kappa^2 L \sigma^{-2}, \, (1-2\varepsilon)\log\left\{(1-\varepsilon)/\varepsilon\right\} L \right\} \geq \zeta_n, 
    \end{equation}
    with $\{\zeta_n\}$ being an arbitrarily diverging sequence.  Noting that $x \mapsto (1-2x) \log\{(1-x)/x\}$ is a decreasing function on $(0,\,0.5]$, to see how \eqref{lowercon} complements \eqref{stn}, we consider the following two regimes. 
    
\textbf{Small $\varepsilon$ regime}. If $\kappa^2\sigma^{-2} < (1-2\varepsilon)\log\left\{(1-\varepsilon)/\varepsilon\right\}$, then \eqref{lowercon} is $\kappa^2L\sigma^{-2} \geq \zeta_n$, which complements the small $\varepsilon$ regime implied by \eqref{stn}, and under which, as implied by Lemma~\ref{lowerbound3}, the lower bound on the localisation rate is of order $\sigma^2\kappa^{-2}$.  This is the same rate without the presence of contamination and this again quantifies the cost of contamination, that is to say, if $\varepsilon$ is small enough, one can hope for a localisation error as if there is no contamination.
        
\textbf{Large $\varepsilon$ regime}. If $\kappa^2\sigma^{-2} \geq (1-2\varepsilon)\log\left\{(1-\varepsilon)/\varepsilon\right\}$, then \eqref{lowercon} is $(1-2\varepsilon)\log\left\{(1-\varepsilon)/\varepsilon\right\} L \geq \zeta_n$.  This complements the large $\varepsilon$ regime implied by \eqref{stn}. The corresponding lower bound, as implied by Lemma~\ref{lowerbound3}, will diverge if $\varepsilon$ tends to $1/2$ at an arbitrary rate. If $\varepsilon$ is bounded away from $1/2$, then regardless of the strength of $\kappa$, the lower bound is of constant order, which is in fact trivial due to the discrete nature of the change points.

\begin{lemma}\label{lowerbound3}
Let $\{Y_i\}_{i = 1}^n$ satisfy Assumption~\ref{modelasmmp} with only one change point and let $P^n_{\kappa,L,\sigma,\varepsilon}$ denote the corresponding joint distribution. Suppose that $\varepsilon_i = \varepsilon$, $i = 1,\dotsc, n$. Consider the class of distribution 
    \begin{gather*}
        \mathcal{P} = \left\{P^n_{\kappa,L,\sigma,\varepsilon}:\,  \min\left\{\kappa^2L\sigma^{-2}, \, (1-2\varepsilon)\log\left\{(1-\varepsilon)/\varepsilon\right\}L \right\} \geq \zeta_n,\, L < n/2\right\},
    \end{gather*}
    where $\{\zeta_n\}$ is any arbitrarily diverging sequence. Then for all $n \in \mathbb{N}$ large enough, it holds that
    \begin{equation*}
       \inf_{\hat{\eta}} \sup_{P\in \mathcal{P}} \mathbb{E}_P\{d_\mathrm{H} (\widehat{\eta}, \, \eta(P))\}  \geq \max\left\{ e^{-1}(1 - \varepsilon)^{-1}\sigma^2\kappa^{-2}, \,\{2(1-2\varepsilon) e \log((1-\varepsilon)/\varepsilon)\}^{-1}\right\},
   \end{equation*}
    where the infimum is over all possible measurable functions of the data and $\eta(P)$ is the set of true change points of $P \in \mathcal{P}$.
\end{lemma}
   
\subsection{The cost of the contamination} \label{compar}

To quantify the cost of the contamination, we compare the results derived above with their counterparts when no contamination presents.

\noindent \textbf{The difficulty of the problem}.  Intuitively speaking, the existence of contamination $H_i$'s increases the difficulty level of detecting the change points in $F_i$'s, and the larger the proportion $\varepsilon$ is the more difficult the problem becomes.  Lemma~\ref{lemma1} details this cost.  When the contamination proportion is small enough to fit in the small $\varepsilon$ regime, no matter how dramatic each contamination distribution $H_i$ is, we are facing a problem with the same difficulty level as if contamination does not exist.  When the contamination proportion is large enough, the difficulty is dominated by the difficulty of a one-sample robust mean estimation problem, no matter how large the minimal spacing is. 

This is indeed interesting, if not surprising, that the difficulty of this robust change point detection problem has a phase transition between the difficulties of change point without contamination and robust estimation without change points.

\noindent \textbf{The accuracy of the localisation}. The localisation error will be affected in the presence of adversarial contamination especially when $\varepsilon$ is large relative to $\kappa/\sigma$ as evidenced by Lemma~\ref{lowerbound3}. Nevertheless, in terms of order, we should aim to achieve the same localisation accuracy as if no contamination exists.

\section{The Adversarially Robust Change Point Detection Algorithm}\label{section3}

In this section, we propose the adversarially robust change point detection method (\textsc{arc}), which borrows the strength from robust estimation and standard change point analysis areas.  

\subsection{Methodology}

\textsc{arc} is a very intuitive algorithm but can achieve nearly optimal results in certain regimes as discussed in Section \ref{optimality}. We scan through the whole time course using the scan statistic $D_h(\cdot)$, which is the absolute difference between two \textsc{rume} estimators \cite{prasad2019unified}. The \textsc{rume} is proposed in the context of one sample robust mean estimation problem based on the idea of shorth estimators \cite{andrews2015robust,Lai_2016} and is shown to be simultaneously optimal under both the heavy tailed and Huber $\varepsilon$-contamination model~\eqref{huber0}.  For completeness, we detail the \textsc{rume} in Algorithm \ref{rumeb1} in the supplementary material. Note that the sample splitting procedure in the \textsc{rume} helps to avoid statistical dependency during the theoretical analysis but may increase the variance of the estimator. 

The scan statistic $D_h(\cdot)$ is a robust variant of the renowned \textsc{cusum} statistic \cite{page1954continuous}, except for two differences.  First, instead of just using the sample average, with the robustness in mind, a robust mean estimator \textsc{rume} is deployed. Second, the \textsc{cusum} statistic takes the difference between two normalised sample means, which can be of two heavily unbalanced samples. However, due to contamination, it is difficult to track the performance of robust mean estimator based on arbitrary sample sizes. Similar concern also arises in the robust clustering problem where each sub-population needs to be represented sufficiently to derive theoretical guarantees~\citep[e.g.][]{cuesta2008robust}.  Despite the ubiquity of such scan statistics in the change point literature, arguably, the most closely-related one would be \cite{niu2012screening}, where $\varepsilon$ is set to be zero and $F_i$'s are assumed to be sub-Gaussian. 

\begin{algorithm}[!h]
	\begin{algorithmic}
		\INPUT $\{Y_i\}_{i=1}^n \subset \mathbb{R}$, $\lambda, h > 0$.
		\State $\mathcal{B} \leftarrow \emptyset$, $\mathcal{C} \leftarrow \emptyset$;
		\For{$j \in \{2h, 2h+1, \ldots, n - 2h\}$}{
		    \State $D_h(j) \leftarrow \Big|\mathrm{RUME} \left(\{Y_i\}_{i=j+1}^{j+2h}\right) - \mathrm{RUME}\left(\{Y_i\}_{i=j-2h+1}^{j}\right)\Big|$;  \Comment{See \cite{prasad2019unified} or Algorithm~\ref{rumeb1}}
		    \If{$j$ is a $4h$-local maximiser of $D_h(j)$} \Comment{See Definition~\ref{def-local-maximisers}}
		        \State $\mathcal{B} \leftarrow \mathcal{B} \cup \{j\}$;
		    \EndIf
		}
		\EndFor
		\For{$l \in \mathcal{B}$}{
		    \If{$|D_{h}(l)| > \lambda$}
		        \State $\mathcal{C} \leftarrow \mathcal{C} \cup \{l\}$;
		    \EndIf
		}
		\EndFor
		\OUTPUT $\mathcal{C}$.
		\caption{Adversarially robust change point detection (\textsc{arc})} \label{alg-main}
	\end{algorithmic}
\end{algorithm}

With $\{D_h(j)\}_{j}$ in hand, we first focus on all the $4h$-local maximisers -- an idea seen in \cite{niu2012screening} when tackling uncontaminated change point detection problems -- defined below.  All $4h$-local maximisers are then thresholded by $\lambda > 0$ to avoid overestimating the number of change points. 

\begin{defn} \label{def-local-maximisers}
For any $h \geq 0$ and $x\in \mathbb{R}$, the interval $(x-h,x+h)$ is called the $h$-neighbourhood of $x$.  We call $x$ an $h$-local maximiser of a function $f(\cdot)$, if $f(x)\geq f(x')$, for any $x'\in (x-h, x+h)$.
\end{defn}

A probably unsatisfactory feature of using the \textsc{rume} is that the proportion of contamination $\varepsilon$ is required as an input.  Unfortunately, this seems to be the bottleneck observed in the majority of optimal robust procedures, including $M$-estimators \citep[e.g.][]{huber1992robust}, truncated means \cite{jerryli19}, and more recently developed high-dimensional robust procedures \citep[e.g.][]{Lai_2016, diakonikolas2017being}. Recently, \citet{chen2016general} proposed a tournament-based procedure for one sample robust mean estimation in the Huber $\varepsilon$-contamination model~(\ref{huber0}) that is adaptive in $\varepsilon$ and it has been used for tuning parameter selection in \cite{prasad2018robust}. We also investigate the empirical performance of our algorithm when using this procedure to select $\varepsilon$ in Section \ref{numerical}. In terms of theory, applying the tournament procedure \cite{chen2016general} requires the knowledge of the density function of $F_i$'s and the i.i.d.~assumption, which are more restrictive than our framework.
Another key tuning parameter in \textsc{arc} is the window width $h$, the theoretical guidance of which is discussed in Sections~\ref{theogurate} and \ref{optimality}, and the practical guidance can be found in Section~\ref{numerical}.%

Lastly, we note that the computational complexity of \textsc{arc} is of order $\mathcal{O}(nh\log(h))$, where the term~$h\log(h)$ is from ranking the data to find the shortest interval involved in \textsc{rume}. Even though the worst case complexity is $\mathcal{O}(n^2\log(n))$, it is still more efficient than the existing methods such as the penalised biweight loss approach with computational complexity $O(n^3)$ (cf.~Corollary 2 in~\cite{fearnhead2019}).

\subsection{Theoretical guarantees}\label{theogurate}

\begin{thm}\label{scp2}
Let $\{Y_i\}_{i = 1}^n$ be an independent sequence satisfying Assumption~\ref{modelasmmp} with $\kappa>0$.  Let~$\{\widehat{\eta}_k\}_{k = 1}^{\widehat{K}}$ be the output of Algorithm~\ref{alg-main}.  

Assume that (i) there exists a sufficiently large absolute constant $C_{\lambda} > 0$ such that $\kappa/\sigma > C_{\lambda} \sqrt{\max\{\varepsilon, \, \log(n)/h\}}$; (ii) there exists an absolute constant $C' > 1$ such that $2\varepsilon'+ 2\sqrt{C'\varepsilon' \log(n)/h} + C'\log(n)/h < 1/2$, where $\varepsilon' = \max\left\{\varepsilon, C'\log(n)/h\right\}$; (iii) the window width satisfies that $h < L/8$; and (iv) the threshold satisfies that $\lambda = C_{\lambda}\sigma\sqrt{\varepsilon'}$.

We then have that there exists an absolute constant $c > 0$ such that 
\[
    \mP\left\{\widehat{K}=K \quad \mbox{and} \quad \max_{k=1}^{\widehat{K}} |\widehat{\eta}_k-\eta_k|\leq 2h \right\} \geq 1 - n^{-c}.
\]
\end{thm}

Theorem~\ref{scp2} shows that under certain conditions, we have consistent estimation on the number of~$f_i$'s change points with $h = o(n)$, in the presence of adversarial attacks.  The localisation error rate is essentially the window width~$h$ and the required signal-to-noise ratio condition is also a function of the window width~$h$. A similar result is obtained in \cite{niu2012screening} with $h\asymp L$ without contamination. Even though we allow $h = O(L)$ in Theorem \ref{scp2}, this is still somewhat unsatisfactory, but we would like to point out that requiring some knowledge of $L$ is a widely observed phenomenon in the change point analysis literature, even when the contamination is absent.  For example, in the wild binary segmentation algorithm \citep[e.g.][]{fryzlewicz2014wild, wang2020}, random intervals are deployed to localise change points.  However, the ideal theoretical performances rely heavily on knowing that the length of these random intervals are of the same order of the minimal spacing $L$~\cite{yu2020review}.

 Given the conditions, we conduct a sanity check on the feasibility of choosing a window width~$h$ implying consistency. Based on Proposition~\ref{prop2} in the supplementary material, we see that the conditions (i)-(iv) in Theorem \ref{scp2} hold if
    \begin{equation*}
        h > \begin{cases}
            \max\left\{10, \,4C_\lambda^2 \sigma^2/\kappa^2 \right\}C'\log(n), & \varepsilon \leq 0.1, \\
        w(\varepsilon)C'\log(n), & 0.1<\varepsilon < 1/4 \min\left\{1, \,\kappa^2/(C_\lambda^2 \sigma^2)\right\},
        \end{cases}
    \end{equation*}
    where $w(\theta) = 1/(1/2 - \sqrt{2\theta(1-2\theta)})$.  The lower bounds on $h$ in both cases are of order $o(n)$.  This means that there exist regimes of $h$ such that $n^{-1}\max_k |\widehat{\eta}_k-\eta_k| \to 0$, which implies the consistency.  

Theorem~\ref{scp2} provides the guarantees of \textsc{arc} when there exists at least one change point, i.e.~$\kappa > 0$.  When $\kappa = 0$, i.e.~there is no change point of $f_i$'s, Algorithm~\ref{alg-main} is still consistent in the following sense.
\begin{cor}\label{nsp}
Let $\{Y_i\}_{i = 1}^n$ be an independent sequence satisfying Assumption~\ref{modelasmmp} with $\kappa = 0$.  Let $\{\widehat{\eta}_k\}_{k = 1}^{\widehat{K}}$ be the output of Algorithm~\ref{alg-main}.  Assume that there exists an absolute constant $C' > 1$ such that $2\varepsilon'+ 2\sqrt{C'\varepsilon' \log(n)/h} + C'\log(n)/h < 1/2$, where $\varepsilon' = \max\left\{\varepsilon, C'\log(n)/h\right\}$; and the thresholding tuning parameter satisfies that $\lambda = C_{\lambda}\sigma\sqrt{\varepsilon'}$.  We then have that there exists an absolute constant $c > 0$ such that $\mP\{\widehat{K} = 0 \} \geq 1 - n^{-c}$.
\end{cor}

\subsection{Optimality of the \textsc{arc} algorithm} \label{optimality}

Recalling the fundamental limits of the adversarially robust change point detection problem, in view of Theorem~\ref{scp2}, \textsc{arc} can be nearly optimal in terms of both the signal-to-noise ratio condition and the localisation rate in certain regimes, with a properly chosen width $h$ depending on the model parameters.  This is indeed restrictive, but generally speaking, robust learning problems suffer from the same restriction. For instance, as we mentioned before, in robust mean estimation problems under the Huber $\varepsilon$-contamination model, the value $\varepsilon$ should be known in order to achieve optimal results in most algorithms and it is impossible to estimate $\varepsilon$ when the contamination distribution is not specified~\cite{chen2016general}.  \citet{padilla2019optimal} studied the nonparametric change point detection problems, which is also a type of robust change point detection problem, and the optimality results thereof rely on the kernel bandwidth~$h$ to be the same order as the minimal signal strength.  In line with our discussions in Section \ref{lowerboundsection}, we consider the following two cases.

\noindent \textbf{Small $\varepsilon$ regime}.  If $\varepsilon \lesssim \log(n)/L$, then with the choice that $h  \asymp L$, the conditions in Theorem~\ref{scp2} become $\kappa\sqrt{L}/\sigma \gtrsim \sqrt{\log(n)}$, which by comparing with the small $\varepsilon$ regime in (\ref{stn}) shows that \textsc{arc} is consistent under the minimal signal-to-noise condition. In this regime,  \textsc{arc} also enjoys the nearly-optimal localisation rate, provided that the minimal segment length satisfies $L \asymp \max\left\{\sigma^2/\kappa^2, 1 \right\}\log(n)$.

\textbf{Large $\varepsilon$ regime}. Consider two cases in this regime. If $\log(n)/L \lesssim \varepsilon \leq 0.1$, then with the choice $h \asymp \log(n)/\varepsilon$, the conditions in Theorem~\ref{scp2} reduce to $\kappa \gtrsim \sigma\sqrt{\varepsilon}$ and $L \gtrsim \log(n)/\varepsilon$, which is the minimal signal-to-noise condition required by comparing to the large $\varepsilon$ regime in (\ref{stn}). In this regime, \textsc{arc} achieves a nearly-optimal localisation rate of $\log(n)/\varepsilon$ if $\varepsilon$ is of constant order. If $0.1< \varepsilon < 1/4 \min\left\{1, \,\kappa^2/(C_\lambda^2 \sigma^2)\right\}$, then with the choice $h \asymp \omega(\varepsilon)\log(n)$, conditions in Theorem~\ref{scp2} reduce to $\kappa \gtrsim \sigma\sqrt{\varepsilon}$ and $L \gtrsim \omega(\varepsilon)\log(n)$, which is the minimal signal-to-noise condition since $\varepsilon$ is of constant order. In this regime, \textsc{arc} achieves a nearly-optimal localisation rate of $\omega(\varepsilon)\log(n)$.

\section{Numerical Results}\label{numerical}

We consider the following competitors to \textsc{arc}: \textsc{pelt}, the Pruned Exact Linear Time method \citep{killick2012optimal}; \textsc{Biweight}, the penalised cost approach based on biweight loss \citep{fearnhead2019}; \textsc{R\_cusum}, recursive application of a Wald-type testing procedure \citep{huvskova1989nonparametric, fearnhead2019} that is shown to be consistent for single change point detection in a robust regression context; and \textsc{R\_UStat}, a $U$-statistic-type robust bootstrap change point test \citep{yu2019robust}, the theory of which is developed for testing change points but not localisation.  \textsc{pelt} serves as a non-robust baseline. As for the other competitors, their theoretical guarantees, if exist, are established against identical heavy-tail contamination at each time point, but not against the potentially, systematic adversarial attacks.

For \textsc{Biweight}, we adopt its default setting denoted as \textsc{Biweight(2)} and a stronger penalty setting \textsc{Biweight}(5). For \textsc{R\_cusum}, we combine it with the wild binary segmentation \cite{fryzlewicz2014wild}, with 500 random intervals and a BIC-type threshold as in \cite{fearnhead2019}. For \textsc{R\_UStat}, we choose the kernel $\nu(x,\,y) = \mathrm{sign}(x-y)$, the bootstrap sample size 100 and the initial block size 250.  

For \textsc{arc}, the choice of $h$ should depend on the applications. In general, we recommend $h = C\log(n)$ where $10\leq C\leq 30$. As for simulation purpose, we fix $h = 170$ and $\lambda = \max\{0.6\sigma,\,8\sigma\varepsilon\}$.  In the simulations, we vary $L$ which serves the purpose of examining the sensitivity of $h$'s choice.  More sensitivity results can be found in Section~\ref{sensitivity} in the supplementary material. Note that the true value of $\sigma$ is used as the input for \emph{all} algorithms in simulation. For real data, due to the suspected short segment length, we consider a range of smaller $h$ and adapt $\lambda$ accordingly to account for the larger estimation error incurred. Different tuning parameter choices for the competitors are also considered. The standard deviations are estimated via the median absolute deviation of the data. For the choice of~$\varepsilon$, we use the true $\varepsilon$ as the input in simulations, and a data-driven automated method based on the tournament procedure considered in \cite{chen2016general}, which views the uncontaminated distributions as Gaussian.  We call \textsc{arc} with automatically-chosen $\varepsilon$ as automated \textsc{arc} (a\textsc{arc}). Due to the variability of \textsc{arc} and a\textsc{arc} inherited from the \textsc{rume} procedure, we run the algorithms over 1000 times on the real data sets and report the mode. See Section~\ref{detailnum} in the supplementary material for further details.

\subsection{Simulations}\label{adversetting}

Recall that the goal is to detect the changes in the means of $F_i$'s with the attacks in the form of $H_i$'s.  We design two settings mimicking two types of adversarial attacks: (i) \textbf{creating spurious change points} and (ii) \textbf{hiding change points}.   Less adversarial settings where the contamination does not rely on the knowledge of change point locations are considered in Section~\ref{lessadv} in the supplementary material.  Throughout, we let $n = 5000$ and $\delta(\cdot)$ be the Dirac measure. 

\textbf{(i)~Creating spurious change points}.  Let $Y_i \sim (1-\varepsilon) \mathcal{N}(0, \sigma^2) + \varepsilon \delta(-3)$, $i \in \{(j-1)M + 1, \ldots, (j-1/2)M\}$, and $Y_i \sim (1-\varepsilon) \mathcal{N}(0,\,\sigma^2) + \varepsilon \delta(3)$, $i \in \{(j-1/2)M + 1, \ldots, jM\}$, $j \in \{1, \ldots, \Delta\}$ and $M = n/\Delta$. An illustration is the left panel in Figure \ref{adverfigure}.   There is no change point on $F_i$'s, but $2\Delta-1$ changes on $\mathbb{E}[Y_i]$ are created through the contamination distributions $H_i$'s.  

\textbf{(ii)~Hiding change points}.  Let $Y_i \sim (1-\varepsilon) \mathcal{N}(0,\,1) + \varepsilon \delta\left(\kappa/(2\varepsilon)\right)$,  $i \in \{(j-1)M + 1, \ldots, (j-1/2)M\}$, and $Y_i \stackrel{i.i.d}{\sim} (1-\varepsilon) \mathcal{N}(\kappa,\,1) + \varepsilon \delta\left\{\kappa (1 - 1/(2\varepsilon))\right\}$, $i \in \{(j-1/2)M + 1, \ldots, jM\}$, $j \in \{1, \ldots, \Delta\}$ and $M = n/\Delta$.  An illustration is the right panel in Figure \ref{adverfigure}, which shows that the change points in $F_i$'s are all offset, in terms of $\mathbb{E}(Y_i)$'s, by the contamination.

\begin{figure}[hbt]
  \centering
        \includegraphics[height=1.3in, width = \textwidth]{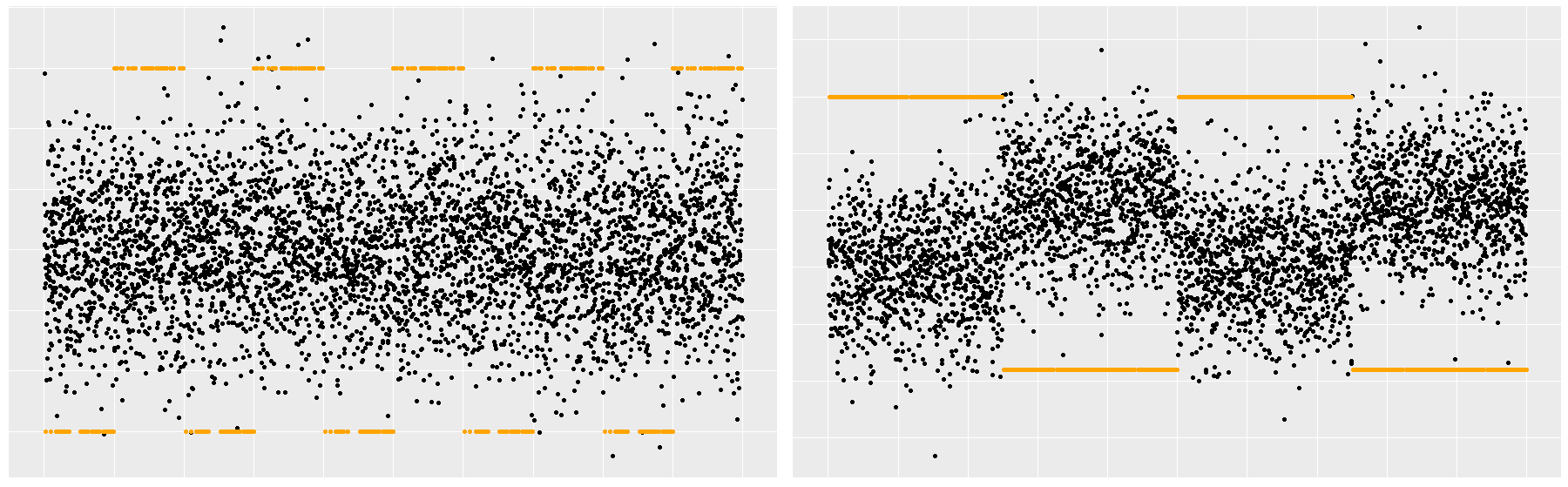}
\caption{The left panel: an illustration of attack setting (i) with  $\varepsilon = 0.1, \Delta = 5$ and $\sigma = 1$. The right panel: an illustration of attack setting (ii) with $\varepsilon = 0.2, \Delta = 2$ and $\kappa = 1.5$.  Black dots are realisations of Gaussian distributions and the orange dots are from the contamination distributions.} 
\label{adverfigure}
\end{figure}

Note that each scenario in \textbf{(i)} is specified by the tuple $\{\varepsilon,\,\Delta,\, \sigma\}$ and we consider 19 different combinations.  The measurement is $\widehat{K} - K$, with $K=0$.  Each scenario in \textbf{(ii)} is specified by the tuple $\{\varepsilon,\,\Delta,\, \kappa\}$ and we consider 12 different combinations.  The measurements are $|\hat{K} - K|$ and $n^{-1}d_H(\hat{\bm{\eta}},\,\bm{\eta})$.  Detailed results can be found in Section~\ref{furthurnumerical} in the supplementary material.  Some representative settings are depicted in Figure~\ref{simulationresult}. 

Overall, if the adversarial attacks are creating spurious change points scenario, when the strength of adversarial attack is strong, e.g.~when $\varepsilon$ is large and when $\Delta$ and $\sigma$ are small, \textsc{arc} and a\textsc{arc} outperform other competitors by detecting fewer spurious change points.  If the adversarial attacks are hiding change points, when $\varepsilon$ is large and the signal $\kappa$ is small, the adversarial noise can fool the competitors such that they consistently output incorrect estimated numbers of change points, while \textsc{arc} and a\textsc{arc} maintain a reasonable performance across all settings.

In terms of localisation errors, in general, we under-perform with respect to the competitors when they can correctly detect the change points, which can be seen as the cost of accuracy when preserving robustness. The observed variability of our results is mainly due to the sample splitting step in the \textsc{rume} procedure and can be improved with a larger sample size.  Moreover, 
we notice that a\textsc{arc} performs competitively comparing with \textsc{arc} with the true value $\varepsilon$ as an input.

\begin{figure}[htb]
   \centering
    \includegraphics[width = \linewidth, height = 1.5in]{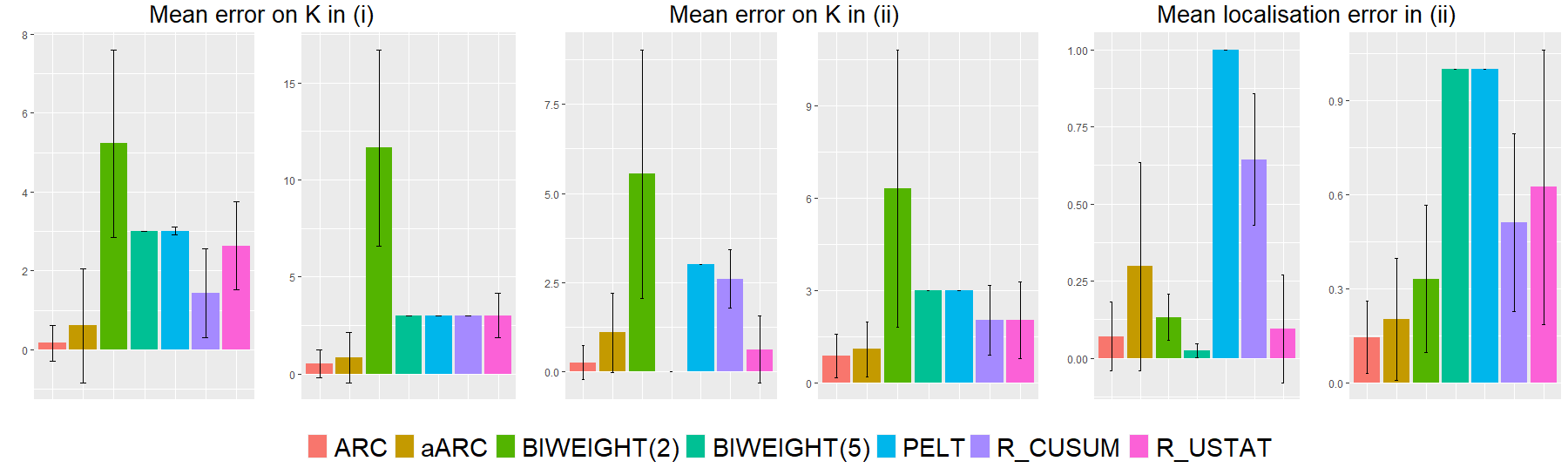}
    \caption{Representative simulation results.  From left to right: $|\widehat{K} - K|$ in setting (i) with $\{\varepsilon,~ \Delta,~\sigma\} = \{0.1,~2,~1\}$; $|\widehat{K}-K|$ in setting (i) with $\{\varepsilon,~ \Delta,~\sigma\} = \{0.2,~2,~1\}$;  $|\widehat{K}-K|$ in setting (ii) with $\{\varepsilon,~ \Delta,~\kappa\} = \{0.1,~2,~0.66\}$; $|\widehat{K}-K|$ in setting (ii) with $\{\varepsilon,~ \Delta,~\kappa\} = \{0.2,~2,~1.2\}$;  $d_H$ in setting (ii) with $\{\varepsilon,~ \Delta,~\kappa\} = \{0.1,~2,~0.66\}$; and $d_H$ in setting (ii) with $\{\varepsilon,~ \Delta,~\kappa\} = \{0.2,~2,~1.2\}$.}
    \label{simulationresult} 
\end{figure}

\subsection{Real data analysis} \label{realdata}

We consider three real data sets in this subsection: the \textbf{well-log data set} that has been extensively studied in the existing literature and two \textbf{PM2.5 index data sets} that are additions to the literature. Section \ref{realdatasupp} in the supplementary material contains additional details on the evaluation method and choices of tuning parameters.

The \textbf{well-log data set} \citep[e.g.][]{wellloglink, fearnhead2019,burg2020evaluation,fukushima2020online} contains 4049 measurements of nuclear magnetic response during the drilling of a well.  We depict the data set and the outputs of a\textsc{arc} in Figure \ref{realdataplot}. With a few isolated observations, the majority seem to behave well. This falls into the regime of the existing robust change point detection methods.  As a result, all competitors considered output very similar results which are omitted in Figure \ref{realdataplot}. 

Consider two \textbf{PM2.5 index data sets} \cite{airquality}: the Beijing PM2.5 index from 15-Apr-2017 to 15-Feb-2021 considered in Section~\ref{sec-intro} and London PM2.5 index from 1-Jan-2014 to 17-Mar-2021.  

In the \textbf{Beijing} data set, \textsc{Biweight}(3), \textsc{R\_cusum} and a\textsc{arc} all detect two change points based on the \textbf{original} data set as shown in the left panel of Figure \ref{teaser}. We then randomly sample 100 points uniformly from the first 1000 data points and change their value to $3\hat{\sigma}$ where $\hat{\sigma}$ is the median absolute deviation of the original data set.  Similarly, we sample 50 points from the remaining data and change their value to $0.5\hat{\sigma}$. After this perturbation, only a\textsc{arc} still detects the previous two points and is robust against the spurious one, while the competitors both wrongly miss the previous ones while picking up the spurious one.

The \textbf{London} data set, compared to the \textbf{well-log data set}, contains way more `outliers' -- in terms of the model~\eqref{huber} -- which can be viewed as either that the distributions of interest $F_i$'s are heavy-tailed while $\varepsilon$ is small or $F_i$'s are well-behaved while $\varepsilon$ is large.   The result of a\textsc{arc} is similar to that of \textsc{R\_cusum}, as they treat the data as if the underlying $\varepsilon$ is large. The \textsc{arc} with an input $\varepsilon = 0.01$ more frequently detects one additional point in the first quarter of the data compared to a\textsc{arc}, which is similar to the result obtained by \textsc{Biweight}(5), as they treat the data as if the underlying $\varepsilon$ is small.

\begin{figure}[htb]
    \centering
    \includegraphics[width = \textwidth,height = 1.5in]{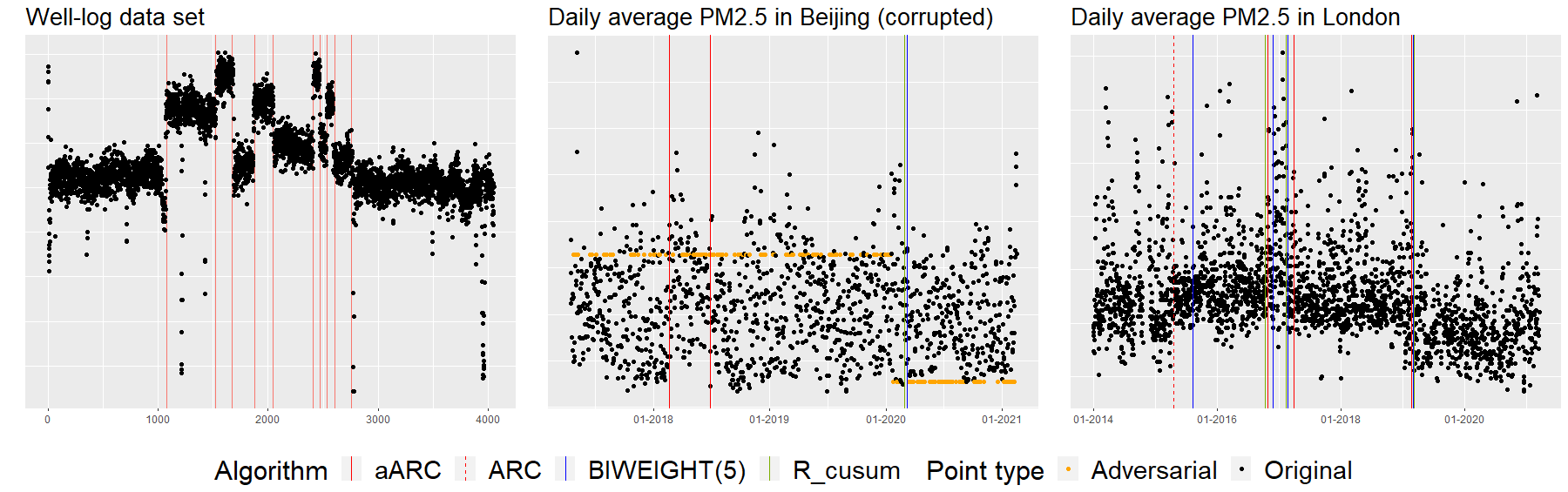}
    \caption{Real data analysis. From left to right: The well-log, Beijing PM2.5 index and London PM2.5 index data sets.  Lines are jittered in the middle and right panels for visualisation purposes.}
    \label{realdataplot}
\end{figure}

\section{Conclusion}\label{sec-conclu}

In this paper, we analysed the change point detection problem under a dynamic extension of the Huber $\varepsilon$-contamination model by allowing the contamination distributions to be different at each point. Under our framework, the adversary can deploy certain attacking strategies that are of interest for the change point detection problem but cannot be modelled within the existing literature. A computationally-efficient algorithm \textsc{arc} that combines the ideas in robust statistics and change point detection literature is shown to be nearly-optimal in terms of both the signal-to-noise ratio condition and the localisation rate when the minimal spacing $L$ is small or when $\varepsilon$ is of constant order. The optimality results developed in this paper are first time shown in the literature, but are still somewhat restrictive.  In order to achieve optimality in the whole parameter space, novel robust estimation and testing techniques are necessary and are on our agenda.

We note that in our framework, even if $\varepsilon$ is of constant order, the optimal localisation error is of the same order compared to the situation when no contamination exists. This is in contrast to the strong contamination model \citep[e.g.][]{pensia2020robust}, where a cluster of outliers are allowed to be created by the adversary, which would unavoidably have a larger impact on the localisation error rate. For example, placing $\varepsilon n$ contaminated points adjacent to a change point would incur a localisation error of $\varepsilon n$ which lead to inconsistent localisation in the sense of (\ref{consistency}). Our future plan includes both considering different contamination models and extending the methodology to more challenging data types.

We conjecture that using optimal and efficient robust estimators for high dimensional data or other complex data types combining with the scanning window idea could lead to similar results as in Theorem \ref{scp2} and Corollary \ref{nsp}.  This is left to our future work.

\begin{ack}

Funding in direct support of this work: DMS-EPSRC EP/V013432/1.

\end{ack}

\newpage
\bibliography{ref}
\newpage
\section*{Supplementary material}
\appendix

\counterwithin{lemma}{section}
\counterwithin{defn}{section}
\counterwithin{prop}{section}
\counterwithin{algorithm}{section}
\counterwithin{remark}{section}
\section{Additional concepts}

We use the two-sided Hausdorff distance to measure the distance between the estimated change points and the true change points.
\begin{defn}[Hausdorff distance]\label{hausdorff}
For any subset $S_1$, $S_2 \subset \mathbb{Z}$, the Hausdorff distance $d_H(S_1,S_2)$ between $S_1$ and $S_2$ is defined to be 
    \[
     \max\left\{\max_{s_1 \in S_1}\min_{s_2\in S_2}|s_1-s_2|, \max_{s_2\in S_2}\min_{s_1\in S_1} |s_1-s_2|\right\},
    \]
    with the convention that 
    \begin{equation*}
    d_H(S_1,S_2) = 
    \begin{cases}
       \infty, &\qquad S_1 = \emptyset \neq S_2 \quad \text{or} \quad S_2 = \emptyset \neq S_1, \\  0, &\qquad S_1 = S_2 = \emptyset.
    \end{cases}
    \end{equation*}
\end{defn}

To compare the performance of our method to a range of other methods studied in \cite{burg2020evaluation}, we consider the following covering metric in Section \ref{detailnum}.
\begin{defn}[Covering metric]\label{covermetric}
For any two partitions $\mathcal{G}$ and $\mathcal{G}'$ of the set $\{1,\ldots,n\}$, the covering metric of partition $\mathcal{G}$ by partition $G'$ is defined as 
\[
C(\mathcal{G}', \mathcal{G}) = \frac{1}{n}\sum_{A \in \mathcal{G}} |A|\max_{A' \in \mathcal{G}'}J(A, A'),
\]

where $|A|$ denotes the cardinality of the set $A$ and $J(A, A')$ is the Jaccard index defined as 
\[
J(A, A') = \frac{|A \cap A'|}{|A\cup A'|}.
\]
\end{defn}

\section{Technical details regarding the Robust Univarite Mean Estimator (RUME)}\label{appendix2}

The \textsc{rume} is proposed and studied in \cite{prasad2019unified}, which is an optimal one-sample robust mean estimator, without contamination.  Our \textsc{arc} algorithm relies on the analysis of \textsc{rume}, with adaptations to allow for different $H_i$'s at each observation.  For completeness, we include all the detailed analysis of \textsc{rume} in this section, with adaptations to the dynamic Huber contamination model studied in this paper.

\begin{algorithm}[htbp]
    \begin{algorithmic}
        \INPUT $\{Z_i\}_{i=1}^{2h} \subset \mathbb{R}$, $0<\varepsilon<1$, $0<\delta<1$
        \State Randomly split $\{Z_i\}_{i=1}^{2h}$ into $\mathcal{Z}$ and $\mathcal{Z}'$ each containing $h$ points;
        \State $\varepsilon \leftarrow \max\left\{\varepsilon, \frac{\log(1/\delta)}{h}\right\}$; 
        \State $D \leftarrow \left\lfloor h\left(1-2\varepsilon-2\sqrt{\varepsilon\frac{\log(1/\delta)}{h}}-\frac{\log(1/\delta)}{h}\right)\right\rfloor$;
        \State $I \leftarrow \emptyset$;
        \For{$j \in \{1,\ldots, h-D\}$}
                \State $I_j \leftarrow Z_{(j+D)}-Z_{(j)}$ \Comment{$Z_{(i)}$ denotes the $i$-th smallest value in $\mathcal{Z}$}
                \State $I \leftarrow I \cup I_j$
        \EndFor
        \State $j^* \leftarrow $ the index of the smallest value in $I$; 
        \State $\hat{I} \leftarrow [Z_{(j^*)},\,Z_{(j^*+D)}]$;
        \State \textsc{RUME}$ \leftarrow \frac{1}{\sum_{i=1}^{h}\mathds{1}\{Z'_i  \in \hat{I}\}}\sum_{i=1}^{h}Z'_i\mathds{1}\{Z'_i\in \hat{I}\}$.
        \OUTPUT{\textsc{RUME}}
        \caption{Robust Univariate Mean Estimation (\textsc{RUME})} \label{rumeb1}
    \end{algorithmic}
\end{algorithm}

Proposition \ref{rume} relies largely on Lemma 3 in \cite{pmlr-v108-prasad20a}, except that we consider model \eqref{huber2}, which allows the contamination distributions to be different for each $Z_i$.  The proof is a minor adaptation from that of Lemma~3 in~\cite{pmlr-v108-prasad20a}.

\begin{prop}[Lemma 3 in \cite{pmlr-v108-prasad20a}]\label{rume}
Suppose $Z_1,\dotsc,Z_{2h}$ are independent random variables with $Z_i$ generated from the distribution
\begin{equation}\label{huber2}
    (1-\varepsilon_i)F_0+\varepsilon_i H_i, \quad, i \in \{1, \ldots, 2h\},
\end{equation}
where $\varepsilon_i \leq \varepsilon$, $F_0$ is any distribution in $\mathbb{R}$ with mean $\mu$ and variance upper bounded by $\sigma^2$ and $H_i$'s are any distributions.  Let 
\[
\varepsilon' = \max\left\{\varepsilon,\,\frac{\log(1/\delta)}{h}\right\}.
\]
Then, if
\begin{equation}\label{A1}
    2\varepsilon'+2\sqrt{\varepsilon'\frac{\log(1/\delta)}{h}}+\frac{\log(1/\delta)}{h} < \frac{1}{2}\;\; \mathrm{and} \;\;
     \delta \leq C'1/h,
\end{equation}
where $C' > 0$ is an absolute constant, then it holds that with probability at least $1-5\delta$, 
\begin{equation}\label{b1}
    |\mathrm{RUME}(\{Z_i\}_{i=1}^{2h})-\mu| \leq C_1\sigma\sqrt{\varepsilon'},
\end{equation}
for some absolute positive constant $C_1$.
\end{prop}

\begin{proof}
Without loss of generality, we can take $\mu = 0$. Let $I^*$ be the interval $(-\sigma/\sqrt{\varepsilon},\,\sigma/\sqrt{\varepsilon})$ and $F_0(I^*)$ denotes the probability that one sample drawn from ($\ref{huber2}$) is distributed according to $F_0$ and lies in $I^*$. If $X \sim F_0$, then by Chebyshev's inequality we have 
\[
\mathbb{P}(|X|>\sigma/\sqrt{\varepsilon}) \leq \varepsilon.
\]
Therefore we have $F_0(I^*) = \mathbb{P}(Z_i \in I^*\, \text{and}\, Z_i \sim F_0) = \mathbb{P}(Z_i \in I^* |Z_i \sim F_0) \mathbb{P}(Z_i \sim F_0) \geq (1-\varepsilon)(1-\varepsilon)\geq 1-2\varepsilon$.

Now let $X_i = \mathds{1}\{Z_i \sim F_0 \,\,\text{and}\,\,Z_i \in  I^* \}$ and $F_0^h(I^*) = \sum_{i=1}^h X_i/h$. Note that $X_i$ is a Bernoulli random variable with success probability $F_0(I^*)$. Therefore, using the Bernstein inequality for bounded random variables (e.g. Theorem 2.8.4 in \cite{vershynin2018high}), we have with probability at least $1-\delta$
\begin{align}
    F_0^h(I^*) - F_0(I^*) &\geq -\sqrt{F_0(I^*)(1-F_0(I^*))}\sqrt{\frac{2\log(1/\delta)}{h}}-\frac{2\log(1/\delta)}{3h} \\
    F_0^h(I^*) &\geq 1-2\varepsilon-\sqrt{2\varepsilon(1-2\varepsilon)\frac{2\log(1/\delta)}{h}}-\frac{2\log(1/\delta)}{3h}\label{I*num},
\end{align}
since $F_0(I^*)(1-F_0(I^*))$ is a decreasing function of $F_0(I^*)$ when $F_0(I^*)>1/2$. Also, note that the Bernstein bound is used here since it improves the Hoeffding bound when the variance of $X_i$ is small.

Let $g_h(2\varepsilon,\delta) = 2\varepsilon+\sqrt{2\varepsilon(1-2\varepsilon)\frac{2\log(1/\delta)}{h}}+\frac{2\log(1/\delta)}{3h}$ and $\hat{I} = [a,b]$ be the shortest interval containing $h(1-g_h(2\varepsilon,\delta))$ points in $\mathcal{Z}$. Since $I^*$ also contains at least $h(1-g_h(2\varepsilon,\delta))$ points due to (\ref{I*num}), we must have
\[
\mathrm{length}(\hat{I})\leq \mathrm{length}(I^*) = 2\sigma/\sqrt{\varepsilon}.
\]
Further, if $g_h(2\varepsilon,\delta) < 1/2$, then both $\hat{I}$ and $I^*$ contain more than half of the data in $\mathcal{Z}$. As a result, these two intervals must intersect and we have 
\begin{equation}\label{eq1}
    |z - \mu |\leq 4\sigma/\sqrt{\varepsilon} \quad \forall z \in \hat{I}.
\end{equation}

Next, we control the error of the final estimator. Let $|\hat{I}| = \sum_{Z_i\in \mathcal{Z}'}\mathds{1}\{Z_i \in \hat{I}\}$ be the number of points from the second sample and lie in $\hat{I}$. Similarly, let $|\hat{I}_{H}|$ and $|\hat{I}_{F_0}|$ denote the number of points that lie in $\hat{I}$ and are \textbf{not} distributed according to $F_0$ (i.e.\ adversarial point) and according to $F_0$ respectively. Note that
\begin{equation*}
    \left|\frac{1}{|\hat{I}|} \sum_{Z_i \in \hat{I}}Z_i  \right| \leq T_1+T_2
\end{equation*}
where 
\begin{equation*}
    T_1 = \left|\frac{1}{|\hat{I}|} \sum_{\substack{Z_i \in \hat{I} \\ Z_i \not\sim F_0}}Z_i \right| \quad \text{and} \quad T_2 = \left|\frac{1}{|\hat{I}|} \sum_{\substack{Z_i \in \hat{I} \\ Z_i \sim F_0}}Z_i  \right|.
\end{equation*}

\textbf{Control of $T_1$}: To control $T_1$, we use (\ref{eq1}) to get 
\[
T_1 \leq \frac{|\hat{I}_{H}|}{|\hat{I}|}\max_{\substack{Z_i \in \hat{I}\\Z_i\not\sim F_0 }}|Z_i| \leq  \frac{|\hat{I}_{H}|}{|\hat{I}|} \frac{4\sigma}{\sqrt{\varepsilon}}.
\]
To bound the ratio $|\hat{I}_{H}|/|\hat{I}|$, notice that the total number of points that are drawn from the adversarial distributions can be controlled by the Bernstein inequality. Therefore, we have with probability $1-\delta$
\begin{equation}\label{H0bound}
    \left|\frac{\hat{I}_{H}}{h}\right| \leq \varepsilon+\sqrt{\varepsilon(1-\varepsilon)}\sqrt{\frac{2\log(1/\delta)}{h}}+\frac{2\log(1/\delta)}{3h} = g_h(\varepsilon,\delta),
\end{equation}
since $|\hat{I}_{H}|$ is less than the total number of points that are drawn from the adversarial distributions. Together, we have with probability at least $1-2\delta$
\begin{equation}\label{t1}
    T_1 \leq \frac{g_h(\varepsilon,\delta)}{1-g_h(2\varepsilon,\delta)}\frac{4\sigma}{\sqrt{\varepsilon}} = C_1 \sigma\sqrt{\varepsilon}
\end{equation}
for some absolute constant $C_1$, where we require $\varepsilon \gtrsim \log(1/\delta)/h$.

\textbf{Control of $T_2$}: To control $T_2$, we write it as 

\[
T_2 = \left|\frac{|\hat{I}_{F_0}|}{|\hat{I}|}\left[\frac{1}{|\hat{I}_{F_0}|}\sum_{\substack{Z_i\in \hat{I}\\Z_i \sim F_0}}Z_i \right]\right| \leq T_{2a} +T_{2b}
\]
where 
\[
T_{2a}=\frac{|\hat{I}_{F_0}|}{|\hat{I}|}\left[\frac{1}{|\hat{I}_{F_0}|}\sum_{\substack{Z_i\in \hat{I}\\Z_i \sim F_0}}Z_i-\mathbb{E}[Z|Z \in \hat{I}, Z \sim F_0]\right] \quad \text{and} \quad T_{2b} = \frac{|\hat{I}_{F_0}|}{|\hat{I}|}\left|\mathbb{E}[Z|Z \in \hat{I}, Z \sim F_0]\right|.
\]
In $T_{2a}$, since conditional on $Z_i \in \hat{I}$, each $Z_i$ is a bounded random variable with $|Z_i-\mathbb{E}(Z_i)|\leq \text{length}(\hat{I}) = 2\sigma/\sqrt{\varepsilon}$ and they are independent of each other, we can again use Bernstein inequality. Using Lemma 15 in \cite{prasad2019unified}, which says for any event $E$ that occurs with probability at least $P(E)$,
\[
\mathbb{E}_{Z\sim F_0}\left[\left(Z-\mathbb{E}[Z|Z\in E]\right)^2|Z\in E\right]\leq \frac{\sigma^2}{\mathbb{P}(E)},
\]
we can obtain an upper bound for the conditional variance of $Z_i$. Denote $F_0(\hat{I})$ to be the probability that $Z_i$ is distributed according to $F_0$ and lies in $\hat{I}$. Then, we have with probability at least $1-\delta$

\begin{equation}\label{t2a}
    T_{2a} \leq \sqrt{\frac{2\sigma^2\log(2/\delta)}{F_0(\hat{I})|\hat{I}_{F_0}|}} + \frac{4\sigma}{\sqrt{\varepsilon}}\frac{\log(2/\delta)}{3|\hat{I}_{F_0}|},
\end{equation}
by the Bernstein inequality. 

For $T_{2b}$, we first notice 
\begin{align*}
   \left|\mathbb{E}_{Z\sim F_0}[Z|Z \not\in \hat{I}]\right| &= \frac{\mathbb{E}_{Z\sim F_0}[Z \mathds{1}_{Z \not\in \hat{I}}]}{F_0(\hat{I}^c)} \\
    & \leq \frac{\sqrt{\mathbb{E}_{Z\sim F_0}[Z^2]F_0(\hat{I}^c)}}{F_0(\hat{I}^c)} \\
    &= \frac{\sigma}{\sqrt{F_0(\hat{I}^c)}},
\end{align*}
where $F_0(\hat{I}^c)$ is the probability that $Z$ is distributed according to $F_0$ but does not lie in $\hat{I}$ and we use the Cauchy-Schwarz inequality in the second line. Combining with fact that 
\[
\left|\mathbb{E}\left[Z|Z\in \hat{I}\right]\right|F_0(\hat{I}) = F_0(\hat{I}^c)\left|\mathbb{E}\left[Z|Z\not\in \hat{I}\right]\right|,
\]
and assuming $F_0(\hat{I}) \geq 1/2$, we have 
\begin{equation}\label{t2b}
    T_{2b} \leq 2\sigma \sqrt{F_0(\hat{I}^c)}.
\end{equation}
 
Combining (\ref{t1}), (\ref{t2a}), and (\ref{t2b}), we get with probability at least $1-3\delta$

\begin{equation}\label{t1t2}
    |\mathrm{RUME}-\mu| \leq C_1 \sigma\sqrt{\varepsilon}+\sqrt{\frac{2\sigma^2\log(2/\delta)}{F_0(\hat{I})|\hat{I}_{F_0}|}} + \frac{4\sigma}{\sqrt{\varepsilon}}\frac{\log(2/\delta)}{3|\hat{I}_{F_0}|}+2\sigma \sqrt{F_0(\hat{I}^c)}.
\end{equation}

To get the claimed bound, we need to study $F_0(\hat{I})$ and $F_0^h(\hat{I}^c) = \frac{|\hat{I}_{F_0}|}{\sum_{Z_i \in \mathcal{Z}'}\mathds{1}_{Z_i \sim F_0}}$. Note that $F_0^h(\hat{I})$ is a sample version of $F_0(\hat{I})$. Let $|\hat{h}_{H}|$ denotes the number of points in $\mathcal{Z}$ which are \textbf{not} drawn from $F_0$ and lie in $\hat{I}$, and $|\hat{h}_{F_0}|$ denotes the number of points in $\mathcal{Z}$ which are drawn from $F_0$ and lie in $\hat{I}$. Note that $|\hat{h}_{H}|$ and $|\hat{I}_{H}|$ have the same distribution, therefore using ($\ref{H0bound}$), we have with probability at least $1-\delta$
\[
\frac{|\hat{h}_{H}|}{h} \leq g_h(\varepsilon,\delta).
\]
Since $|\hat{h}_{H}|+|\hat{h}_{F_0}| = h(1-g_h(2\varepsilon,\delta))$, we have with probability at least $1-\delta$
\begin{equation}\label{pn}
    |\hat{h}_{F_0}| \geq h\left(1-g_h(2\varepsilon,\delta)-g_h(\varepsilon,\delta)\right).
\end{equation}
Note that $|\hat{I}_{F_0}|$ and $|\hat{h}_{F_0}|$ also have the same distribution. Therefore, with probability $1-\delta$, 
\begin{equation}\label{hatI}
    |\hat{I}_{F_0}| \geq h(1-g_h(2\varepsilon,\delta)-g_h(\varepsilon,\delta)) = C_2 h
\end{equation}
for some constant $C_2$, where we require $\varepsilon \gtrsim \log(1/\delta)/h$. Equation (\ref{pn}) implies that 
\[
F_0^h(\hat{I}^c) = \frac{|\hat{I}_{F_0}|}{\sum_{Z_i \in \mathcal{Z}'}\mathds{1}_{Z_i \sim F_0}} \geq 1-g_h(2\varepsilon,\delta)-g_h(\varepsilon,\delta).
\]
Consequently, we have 
\begin{equation}
    F_0^h(\hat{I^c})\leq g_h(2\varepsilon,\delta)+g_h(\varepsilon,\delta) = 3\varepsilon + (\sqrt{\varepsilon(1-\varepsilon)}+\sqrt{2\varepsilon(1-2\varepsilon)})\sqrt{\frac{2\log(1/\delta)}{h}} + \frac{4\log(1/\delta)}{3h}.
\end{equation}
Provided that $\varepsilon \lesssim \log(1/\delta)/h$, we have 
\begin{equation}\label{epl}
    F_0^h(\hat{I^c})\leq 3\varepsilon + C_3 \frac{\log(1/\delta)}{h}.
\end{equation}
Using the relative deviation lemma from empirical process theory \citep[e.g.\ Theorem 7 in][]{bousquet2003introduction}, we can finally bound $F_0(\hat{I}^c)$ as 
\begin{equation}\label{vc}
    F_0(\hat{I}^c) \leq F_0^h(\hat{I}^c)+2\sqrt{F_0^h(\hat{I}^c)\frac{\log(S_{\mathcal{F}}(2h))+\log(4/\delta)}{h}}+4 \frac{\log(S_{\mathcal{F}}(2h))+\log(4/\delta)}{h},
\end{equation}
with probability at least $1-\delta$. Since the VC dimension for intervals in $\mathbb{R}$ is 2, we have $S_{\mathcal{F}}(2h) \leq (2h+1)^2$ by the Sauer-Shelah Lemma (e.g.~Theorem 8.3.16 in \cite{vershynin2018high}). Substituting the upper bound of $S_{\mathcal{F}}(2h)$ and $F_0^h(\hat{I^c})$ into equation (\ref{vc}) and using the fact that $\sqrt{ab}\leq a+b$ for any $a,b\geq 0$, we get with probability at least $1-2\delta$
\begin{equation}\label{empi}
    F_0(\hat{I}^c) \leq 9\varepsilon + C_6 \frac{\log(1/\delta)}{h}+12\frac{\log(2h+1)}{h} = 9\varepsilon + C_6 \frac{\log(1/\delta)}{h}+C_7\frac{\log(h)}{h}.
\end{equation}

Combining (\ref{t1t2}), (\ref{hatI}), and (\ref{empi}), we have that with probability at least $1-5\delta$
\[
|\mathrm{RUME}-\mu|\leq C_1 \sigma\sqrt{\varepsilon}+C_2 \sigma\sqrt{\frac{\log(h)}{h}}+C_3 \sqrt{\frac{\log(1/\delta)}{h}}+C_4\sigma\frac{\log(1/\delta)}{\sqrt{\varepsilon}h}.
\]

The claimed bound follows by letting $$\varepsilon' = \max\left\{\varepsilon, \frac{\log(1/\delta)}{h}\right\}$$ and choosing $\delta \lesssim 1/h$. 
\end{proof}

\section{Proofs of the results in Section \ref{lowerboundsection}}\label{appendix1}

In this section, we prove Lemmas~\ref{lemma1} and \ref{lowerbound3}. Throughout this section, we use $\delta(x)$ to denote the Dirac measure at point $x$ and consider $\varepsilon_1 = \dotsc = \varepsilon_n = \varepsilon$ in Assumption \ref{modelasmmp}. We prove Lemma \ref{lemma1} by considering three sub-problems in Lemmas~\ref{B1}, \ref{sublemma} and \ref{lowerbound1}, respectively. Note that Lemma~\ref{B1} is from~\cite{wang2020} and it quantifies the difficulty of the change point detection problem without any contamination. The construction in the proof of Lemma \ref{lowerbound1} also appears in \cite{Lai_2016}, but we formulate it here formally in Le Cam's framework. Lemma~\ref{lowerbound3} is proved by considering two sub-problems in Lemmas~\ref{B3} and \ref{B4}.

     \begin{proof}[Proof of Lemma \ref{lemma1}] 
     
     Let $s = \max\left\{8\varepsilon, \log(n)(1-2\varepsilon), 4\varepsilon(1-2\varepsilon)L\right\}$. To prove Lemma \ref{lemma1}, it is sufficient to prove that the claim 
     \begin{equation}\label{B.1}
         \inf_{\hat{\eta}} \sup_{P\in \mathcal{P}} \mathbb{E}_P(d_H(\hat{\bm{\eta}},\,\bm{\eta}(P))) \geq \frac{n}{8}
     \end{equation}
     holds in three cases. First, note that when $s = \log(n)(1-2\varepsilon)$, the claim (\ref{B.1}) follows from Lemma \ref{B1} below. Next, we show that the claim (\ref{B.1}) holds when $s = 8\varepsilon$ in Lemma \ref{sublemma} below. To conclude the proof, we show that the claim (\ref{B.1}) holds when $s = 4\varepsilon(1-2\varepsilon)L$ in Lemma  \ref{lowerbound1} below.
   \end{proof}
   
   \begin{lemma}[Lemma 1 in \cite{wang2020}]\label{B1}
   Let $\{Y_i\}_{i=1}^n$ satisfy Assumption \ref{modelasmmp} with $\varepsilon = 0$ and $F_1,\dotsc,F_n$ being sub-Gaussian random variables. Let $P_{\kappa,\sigma,L}^n$ denote the corresponding joint distribution. For any $0<c<1$, consider the class of distributions
   \begin{equation*}
       \mathcal{P}^n_c = \left\{P_{\kappa,\sigma,L}^n: L = \min\left\{c\frac{\log(n)}{\kappa^2/\sigma^2},\,\frac{n}{4}\right\} \right\}.
   \end{equation*}
   Then, there exists a $n(c)$, which depends on $c$, such that, for all $n$ larger than $n(c)$, it holds that 
   \begin{equation*}
       \inf_{\hat{\bm{\eta}}} \sup_{P\in \mathcal{P}^n} \mathbb{E}_P(d_H(\hat{\bm{\eta}},\,\bm{\eta}(P))) \geq \frac{n}{8}
   \end{equation*}
   where the infimum is over all estimators $\bm{\hat{\eta}}$ of the change point locations and $\bm{\eta}(P)$ denotes the change point location of $P\in \mathcal{P}$.
\end{lemma}

   \begin{lemma}\label{sublemma}
    Let $Y_1 \dotsc,Y_n$ be a time series satisfying Assumption 1 with only one change point and let $P^n_{\kappa,L,\sigma,\varepsilon}$ denote the corresponding joint distribution. Consider the class of distribution 
    \[
    \mathcal{P} = \left\{P^n_{\kappa,L,\sigma,\varepsilon} : \frac{\kappa^2 L}{\sigma^2} < \frac{8\varepsilon}{1-2\varepsilon}, L \leq \floor*{\frac{n}{4}}\right\},
    \]
    then  \begin{equation*}
       \inf_{\hat{\eta}} \sup_{P\in \mathcal{P}} \mathbb{E}_P(d_H(\hat{\bm{\eta}},\,\bm{\eta}(P))) \geq \frac{n}{4}
   \end{equation*}
   where the infimum is over all estimators $\hat{\eta}$ of the change point locations and $\bm{\eta}(P)$ is the true change point of $P \in \mathcal{P}$.
    \end{lemma}
    \begin{proof}[Proof of Lemma \ref{sublemma}]
      Without loss of generality, suppose $n/L = c(n)$, where $c(n)$ is an integer that is allowed to depend on $n$. Denote the density of $\mathcal{N}(u_0,\,\sigma^2 I)$ by $\phi_0$ and the density of $\mathcal{N}(u_k,\,\sigma^2I)$ by $\phi_\kappa$, where $u_0 \in \mathbb{R}^L$ is a vector with all entries being $0$, $u_\kappa \in \mathbb{R}^L$ is a vector with all entries being $\kappa$, and $I$ is the identity matrix of dimension $L \times L$.  Let $\mathds{1}_{\phi_\kappa>\phi_0}$ be the indicator function, i.e. $\mathds{1}_{\phi_\kappa>\phi_0}(x) = 1$ if $\phi_\kappa(x)>\phi_0(x)$ and $\mathds{1}_{\phi_2>\phi_1}(x) = 0$ otherwise. Let $\widetilde{P}$ denote the joint distribution of $\{Y_i\}_{i=1}^n$ with one change point at $L \leq n/4$ such that 
    \[
    Y_i \sim (1-\varepsilon)F_i+ \varepsilon H_i
    \]
    where $F_1 = F_2=\dotsc =F_L= \mathcal{N}(0,\,\sigma^2)$ and $F_{L+1}=\dotsc = F_n = \mathcal{N}(\kappa,\,\sigma^2)$. For the contamination distributions, we choose the joint distribution of $\{H_{i}\}_{i=1}^L$ in $\widetilde{P}$ to have the density 
    $$\frac{1-\varepsilon}{\varepsilon}(\phi_\kappa-\phi_0)\mathds{1}_{\phi_\kappa>\phi_0},$$
    and $\{H_{i}\}_{i=jL+1}^{(j+1)L}$, in $\widetilde{P}$ to have joint distribution with density 
    $$\frac{1-\varepsilon}{\varepsilon}(\phi_0-\phi_\kappa)\mathds{1}_{\phi_0>\phi_\kappa},$$
    for $j = 1, \dotsc, c(n)-1$. 
    
    Similarly, let $\widetilde{Q}$ denote the joint distribution of $\{Y_i'\}_{i=1}^n$ with one change point at $L'  = n-L$ such that 
    \[
    Y_i \sim (1-\varepsilon)F'_i+ \varepsilon H'_i
    \]
    where $F'_1 = F'_2=\dotsc =F'_{L'}= \mathcal{N}(0,\,\sigma^2)$ and $F'_{L'+1}=\dotsc = F'_n = \mathcal{N}(\kappa,\,\sigma^2)$. For contamination distributions, we choose $\{H_i'\}_{i=1}^{L'}$ in the same way as $\{H_i\}_{i = L+1}^n$, and $\{H_i'\}_{i=L'+1}^n$ in the same way as $\{H_i\}_{i = 1}^L$
    
    Finally, choose the contamination proportion $\varepsilon$ such that  
    \[
    \mathrm{TV}\left(\mathcal{N}(u_0,\,\sigma^2 I),\, \mathcal{N}(u_\kappa,\,\sigma^2I)\right) = \frac{\varepsilon}{1-\varepsilon}.
    \]
    The described joint distributions $\widetilde{P}$ and $\widetilde{Q}$ are indeed belong to $\mathcal{P}$. This can be checked by using the Hellinger distance $H(\cdot,\,\cdot)$, as a lower bound for the total variation distance  
    \begin{align*}
        \mathrm{TV}\left(\mathcal{N}(u_0,\,\sigma^2 I),\,\mathcal{N}(u_k,\,\sigma^2I)\right) &\geq H^2\left(\mathcal{N}(u_0,\,\sigma^2 I),\,\mathcal{N}(u_k,\,\sigma^2I)\right) \\
        & = 1- \int_{\mathbb{R}^L} \sqrt{\phi_0(x)\phi_k(x)}\, \mathrm{dx}\\
        & = 1- \exp\left(-\frac{1}{8\sigma^2}(u_\kappa-u_0)^T(u_\kappa-u_0)\right)\\
        & = 1-\exp\left(-\frac{1}{8}\frac{\kappa^2L}{\sigma^2}\right)
    \end{align*}
    Therefore, we have 
    \[
     1-\exp\left(-\frac{1}{8}\frac{\kappa^2L}{\sigma^2}\right) \leq \frac{\varepsilon}{1-\varepsilon},
    \]
    which is equivalent to 
    \begin{equation}\label{left}
         \frac{\kappa^2 L}{\sigma^2} \leq 8\log \left(1+\frac{\varepsilon}{1-2\varepsilon}\right) \leq \frac{8 \varepsilon}{1-2\varepsilon}.
    \end{equation}

    Finally, we check that under the construction of $\widetilde{P}$ and $\widetilde{Q}$ we have $\widetilde{Q} = \widetilde{P}$. Notice the identity that 
    \begin{equation}\label{indenti}
         \phi_0+(\phi_\kappa-\phi_0)\mathds{1}_{\phi_\kappa>\phi_0} = \phi_\kappa+(\phi_0-\phi_\kappa)\mathds{1}_{\phi_0>\phi_\kappa}.
    \end{equation}
    The joint distribution of $Y_1,\dotsc,Y_L$ has density 
    \begin{equation}\label{d1}
        (1-\varepsilon)\phi_0+(1-\varepsilon)(\phi_\kappa-\phi_0)\mathds{1}_{\phi_\kappa>\phi_0},
    \end{equation}
    while the joint distribution of $Y'_1,\dotsc,Y'_L$ has density
    \begin{equation}\label{d2}
        (1-\varepsilon)\phi_\kappa+(1-\varepsilon)(\phi_0-\phi_\kappa)\mathds{1}_{\phi_0>\phi_\kappa}.
    \end{equation}
    Therefore, the two joint distributions agree as a result of equation (\ref{indenti}). Similarly, the joint distribution of $Y_{jL+1},\dotsc, Y_{(j+1)L}$ has density as in equation (\ref{d2}), for $j=1,\dotsc,c(n)-1$, and the joint distribution of $Y'_{jL+1},\dotsc, Y'_{(j+1)L}$ has density as in equation (\ref{d1}), for $j=1,\dotsc,c(n)-1$. As a result, we can conclude $\widetilde{Q} = \widetilde{P} \in \mathcal{P}$.

    Note that $d_H(\bm{\eta}(\widetilde{P}),\,\bm{\eta}(\widetilde{Q})) \geq n/2$ by construction, therefore by the Le Cam Lemma \citep[e.g.][]{yu1997assouad}, we have 
    
   \[ \inf_{\hat{\bm{\eta}}} \sup_{P\in \mathcal{P}} \mathbb{E}_P(d_H(\hat{\bm{\eta}},\,\bm{\eta}(P))) \geq \frac{n}{4}\{1-\text{TV}(\widetilde{P},\widetilde{Q})\} = \frac{n}{4}.
   \]
    \end{proof}
 
 \begin{lemma}\label{lowerbound1}
    Let $Y_1 \dotsc,Y_n$ be a time series satisfying Assumption \ref{modelasmmp} with only one change point and let $P^n_{\kappa,L,\sigma,\varepsilon}$ denote the corresponding joint distribution. Consider the class of distribution 
    \[
    \mathcal{P} = \left\{P^n_{\kappa,L,\sigma,\varepsilon} : \frac{\kappa}{\sigma} < 2\sqrt{\varepsilon}, L \leq \floor*{\frac{n}{4}} \right\},
    \]
    then  \begin{equation*}
       \inf_{\hat{\eta}} \sup_{P\in \mathcal{P}} \mathbb{E}_P(d_H(\hat{\bm{\eta}},\,\bm{\eta}(P))) \geq \frac{n}{4}
   \end{equation*}
   where the infimum is over all estimators $\hat{\bm{\eta}}$ of the change point locations and $\bm{\eta}(P)$ is the true change point of $P \in \mathcal{P}$.
    \end{lemma}
    
\begin{proof}[Proof of Lemma \ref{lowerbound1}]
    Let $\widetilde{P}$ denote the joint distribution of $\{Y_i\}_{i=1}^n$ with one change point at $L \leq n/4$ such that
    \[
    Y_i \sim (1-\varepsilon)F_i+ \varepsilon H_i
    \]
    where $F_1 = F_2=\dotsc =F_L=\delta(0)$ and $F_{L+1}=\dotsc = F_n$ have the following distribution 
    
    \[   \left\{
\begin{array}{ll}
      F_i = \delta(0) \quad \text{with probability} \quad 1-\varepsilon \\
      F_i = \delta(\kappa/\varepsilon) \quad \text{with probability} \quad \varepsilon \\
\end{array} 
\right. \]
for $i = L+1,\dotsc,n$. Under this setting, we have $f_1 = f_2 = \dotsc = f_L =0$ and $f_{L+1} = \dotsc = f_n = \kappa$. For the outlier distributions, we choose $H_1 = \dotsc = H_L = \delta(\kappa/\varepsilon)$ while $H_{L+1} = \dotsc = H_n$ have the same distribution as $F_n$. Note that by this construction, we have $\{Y_i\}_{i=1}^n$ are independent identically distributed. 

Similarly, let $\widetilde{Q}$ denote the joint distribution of $\{Y'_i\}_{i=1}^n$ with one change point at $L' \geq 3n/4$ such that 
 \[
    Y_i' \sim (1-\varepsilon)F_i'+ \varepsilon H_i'
 \]
where we choose $F'_1 = \dotsc = F'_{L}$ to be the same as $F_1$ and $F'_{L+1} = \dotsc = F'_{n}$ to be the same as $F_n$. For the outlier distributions, we choose $H_1' = \dotsc = H_{L'}'$ to be the same as $H_1$ while $H_{L'+1}' = \dotsc, H_n'$ to be the same as $H_n$.  

     Note that under this construction, we have $\widetilde{P} = \widetilde{Q} \in \mathcal{P}$, since $\{F_i\}_{i=1}^L$ have variance $0$, and $\{F_i\}_{i=L+1}^{n}$ have variance 
     \[
     \kappa^2(1/\varepsilon-1),
     \]
     which is less equal to $\sigma^2$ under 
     \[
     \frac{\kappa}{\sigma} \leq \sqrt{\frac{2\varepsilon}{1-\varepsilon}} < 2\sqrt{\varepsilon}
     \]
    Since $d_H(\bm{\eta}(\widetilde{P}),\,\bm{\eta}(\widetilde{Q}))\geq n/2$ by construction, therefore by the Le Cam lemma  \citep[e.g.][]{yu1997assouad}, we have 
    
   \[ \inf_{\hat{\eta}} \sup_{P\in \mathcal{P}} \mathbb{E}_P(d_H(\hat{\bm{\eta}},\,\bm{\eta}(P))) \geq \frac{n}{4}\{1-\text{TV}(\widetilde{P},\widetilde{Q})\} = \frac{n}{4}.
   \]
    \end{proof}
    
 \begin{proof}[Proof of Lemma \ref{lowerbound3}] 
 To prove Lemma \ref{lowerbound3}, we consider two classes of distributions
  \begin{gather*}
        \mathcal{P}_1 = \left\{P^n_{\kappa,L,\sigma,\varepsilon} : \frac{\kappa^2 L}{\sigma^2} \geq \zeta_n,\;\; L < \frac{n}{2}\right\},\\
        \mathcal{P}_2 = \left\{P^n_{\kappa,L,\sigma,\varepsilon} : (1-2\varepsilon)\log\left(\frac{1-\varepsilon}{\varepsilon}\right)L\geq \zeta_n, \;\;L < \frac{n}{2}\right\}
    \end{gather*}
 Notice that $\mathcal{P} = \mathcal{P}_1 \cap \mathcal{P}_2$. Therefore, the proof can be completed in two steps. First, we show in Lemma \ref{B3} below that 
 \[
 \inf_{\hat{\bm{\eta}}} \sup_{P\in \mathcal{P}_1} \mathbb{E}_P(d_H(\hat{\bm{\eta}},\,\bm{\eta}(P))) \geq \frac{\sigma^2}{\kappa^2}\frac{e^{-1}}{1-\varepsilon}
 \]
 for all $n$ large enough. Then, we show in Lemma \ref{B4} below that 
 \[
  \inf_{\hat{\bm{\eta}}} \sup_{P\in \mathcal{P}_2} \mathbb{E}_P(d_H(\hat{\bm{\eta}},\,\bm{\eta}(P))) \geq \frac{1}{2(1-2\varepsilon)}\frac{e^{-1}}{\log((1-\varepsilon)/\varepsilon)}. 
 \]
  for all $n$ large enough. Thus, the claim follows. 
 \end{proof}
 
 \begin{lemma}\label{B3}
    Let $Y_1 \dotsc,Y_n$ be a time series satisfying Assumption \ref{modelasmmp} with only one change point and let $P^n_{\kappa,L,\sigma,\varepsilon}$ denote the corresponding joint distribution. Consider the class of distribution 
    \[
    \mathcal{P} = \left\{P^n_{\kappa,L,\sigma,\varepsilon} : \frac{\kappa^2 L}{\sigma^2} \geq \zeta_n, L < \frac{n}{2}\right\},
    \]
    for any sequence $\{\zeta_n\}$ such that $\lim_{n\rightarrow\infty}\zeta_n = \infty$. Then for all $n$ large enough, it holds that
    \begin{equation*}
       \inf_{\hat{\bm{\eta}}} \sup_{P\in \mathcal{P}} \mathbb{E}_P(d_H(\hat{\bm{\eta}},\,\bm{\eta}(P))) \geq \frac{\sigma^2}{\kappa^2}\frac{e^{-1}}{1-\varepsilon}
   \end{equation*}
   where the infimum is over all estimators $\hat{\bm{\eta}}$ of the change point locations and $\bm{\eta}(P)$ is the true change point of $P \in \mathcal{P}$.
\end{lemma}
    \begin{proof}[Proof of Lemma \ref{B3}]
    Let $P_0$ denote the joint distribution of independent random variables $\{Y_i\}_{i=1}^n$ where each $Y_i$ has distribution
    \[
     (1-\varepsilon) F^0_i+\varepsilon H^0_i.
    \]
    Let 
    \[
    F_1^0 = F_2^0 = \dotsc = F_L^0 = \mathcal{N}(0,\,\sigma^2) \quad \text{and} \quad F_{L+1}^0= F_{L+2}^0 = \dotsc = F_n^0 = \mathcal{N}(\kappa,\,\sigma^2).
    \]
    Similarly, let $P_1$ be the joint distribution of independent random variables $\{Z_i\}_{i=1}^n$ where $Z_i$ has distribution
    \[
    (1-\varepsilon) F^1_i+\varepsilon H^1_i.
    \]
    Let
     \[
    F_1^1 = F_2^1 = \dotsc = F_{L+\Delta}^1 = \mathcal{N}(0,\sigma^2) \quad \text{and} \quad F_{L+\Delta+1}^1= F_{L+2}^1 = \dotsc = F_n^1 = \mathcal{N}(\kappa,\sigma^2),
    \]
    where $\Delta$ is an integer no larger than $n-1-L$. For the adversarial noise distribution, we choose $H^0_1 = \dotsc = H^0_n = H^1_{1} = \dotsc= H^1_n$, i.e.\ the contamination distribution is the same across time and is the same for $P_0$ and $P_1$.
    
    By the Le Cam Lemma \citep[e.g.][]{yu1997assouad} and Lemma 2.6 in \cite{tsybakov2008introduction}, we have 
    \[
    \inf_{\hat{\eta}} \sup_{P\in \mathcal{P}} \mathbb{E}_P(d_H(\hat{\bm{\eta}},\,\bm{\eta}(P))) \geq \Delta(1-\mathrm{TV}(P_0,P_1)) \geq \frac{\Delta}{2}\exp(-\mathrm{KL}(P_0||P_1)).
    \]
    Since both $P_0$ and $P_1$ are product measures, it holds that 
    \[
    \mathrm{KL}(P_0||P_1) = \sum_{i=L+1}^{L+\Delta}\mathrm{KL}((1-\varepsilon) F^0_i+\varepsilon H^0_i||(1-\varepsilon) F^1_i+\varepsilon H^1_i).
    \]
    Using convexity of the KL divergence (e.g.\ Lemma 1 in \cite{do2003fast}), we have 
    \begin{align*}
        \mathrm{KL}(P_0||P_1) &\leq  \sum_{i=L+1}^{L+\Delta} \left((1-\varepsilon)\mathrm{KL}\left(\mathcal{N}(\kappa,\sigma^2)||\mathcal{N}(0,\sigma^2)\right)+\varepsilon\mathrm{KL}\left(H^0_i||H^1_i\right)\right) \\
        &=\Delta (1-\varepsilon)\frac{\kappa^2}{2\sigma^2}+ \varepsilon\sum_{i=L+1}^{L+\Delta}\mathrm{KL}\left(H^0_i||H^1_i\right) \\
        & = \Delta (1-\varepsilon)\frac{\kappa^2}{2\sigma^2}.
    \end{align*}
    since $H^0_i = H^1_i$, for $i = L+1,\dotsc,L+\Delta$. Hence, we have 
    \begin{equation}\label{kllower}
        \inf_{\hat{\eta}} \sup_{P\in \mathcal{P}} \mathbb{E}_P(d_H(\hat{\bm{\eta}},\,\bm{\eta}(P))) \geq \frac{\Delta}{2}\exp\left(-\Delta (1-\varepsilon)\frac{\kappa^2}{2\sigma^2}\right).
    \end{equation}
    
    Next, set $\Delta = \min\{{2\sigma^2/(1-\varepsilon)\kappa^2}, n-1-L\}$. Using the assumption that 
    \[
    \frac{\kappa^2L}{\sigma^2} \geq \zeta_n.
    \]
    where $\zeta_n$ is a diverging sequence, and
    \[
    \zeta_n > \frac{n/2}{n-1-L} \geq \frac{L}{n-1-L},
    \]
    for all $n$ large enough, we have 
    \[
    \frac{(1-\varepsilon)\kappa^2L}{2\sigma^2} > \frac{\kappa^2L}{4\sigma^2} > \frac{L}{n-1-L},
    \]
    for all $n$ large enough. Therefore, it must hold that $\Delta = {2\sigma^2/(1-\varepsilon)\kappa^2}$ for $n$ large enough and the claimed bound follows from (\ref{kllower}).
      \end{proof}

 \begin{lemma}\label{B4}
    Let $Y_1 \dotsc,Y_n$ be a time series satisfying Assumption 1 with only one change point and let $P^n_{\kappa,L,\sigma,\varepsilon}$ denote the corresponding joint distribution. Consider the class of distribution 
    \[
    \mathcal{P} = \left\{P^n_{\kappa,\delta,\sigma,\varepsilon} : (1-2\varepsilon)\log\left(\frac{1-\varepsilon}{\varepsilon}\right)L\geq \zeta_n, L < \frac{n}{2}\right\},
    \]
    for any sequence $\{\zeta_n\}$ such that $\lim_{n\rightarrow\infty}\zeta_n = \infty$. Then for all $n$ large enough, it holds that
    \begin{equation*}
       \inf_{\hat{\eta}} \sup_{P\in \mathcal{P}} \mathbb{E}_P(d_H(\hat{\bm{\eta}},\,\bm{\eta}(P))) \geq \frac{1}{2(1-2\varepsilon)} \frac{e^{-1}}{\log((1-\varepsilon)/\varepsilon)}
   \end{equation*}
   where the infimum is over all estimators $\hat{\bm{\eta}}$ of the change point locations and $\bm{\eta}(P)$ is the true change point of $P \in \mathcal{P}$.
\end{lemma}
    
    \begin{proof}[Proof of Lemma \ref{B4}]
    Let $P_0$ denote the joint distribution of $\{Y_i\}_{i=1}^n$ where each $Y_i$ has distribution 
    \[
    (1-\varepsilon) F^0_i+\varepsilon H^0_i.
    \]
    Let 
    \[
    F_1^0 = F_2^0 = \dotsc = F_L^0 = \delta(0) \quad \text{and} \quad F_{L+1}^0= F_{L+2}^0 = \dotsc = F_n^0 = \delta(\kappa).
    \]
    The outlier distributions are chosen as 
    \[
    H_1^0 = H_2^0 = \dotsc = H_L^0 = \delta(\kappa) \quad \text{and} \quad H_{L+1}^0= H_{L+2}^0 = \dotsc = H_n^0 = \delta(0)
    \]
    Similarly, let $P_1$ be the joint distribution of random variables $\{Z_i\}_{i=1}^n$ where each $Z_i$ has distribution
    \[
    (1-\varepsilon) F^1_i+\varepsilon H^1_i.
    \]
    Let
     \[
    F_1^1 = F_2^1 = \dotsc = F_{L+\Delta}^1 = \delta(0) \quad \text{and} \quad F_{L+\Delta+1}^1= F_{L+2}^1 = \dotsc = F_n^1 = \delta(\kappa),
    \]
    where $\Delta$ is an integer no larger than $n-1-L$. The outlier distributions are chosen as 
    \[
    H_1^0 = H_2^0 = \dotsc = H_{L+\Delta}^0 = \delta(\kappa) \quad \text{and} \quad H_{L+\Delta+1}^0= H_{L+2}^0 = \dotsc = H_n^0 = \delta(0)
    \]
     
    By the Le Cam Lemma \citep[e.g.][]{yu1997assouad} and Lemma 2.6 in \cite{tsybakov2008introduction}, we have 
    \[
    \inf_{\hat{\eta}} \sup_{P\in \mathcal{P}} \mathbb{E}_P(|\hat{\eta}-\eta(P)|) \geq \frac{\Delta}{2}(1-\mathrm{TV}(P_0,P_1)) \geq \frac{\Delta}{2}\exp(-\mathrm{KL}(P_0,P_1)).
    \]
    Note that, for $i = L+1,\dotsc,L+\Delta$, we have
    \begin{align*}
        \mathrm{KL}\left((1-\varepsilon) F^0_i+\varepsilon H^0_i, (1-\varepsilon) F^1_i+\varepsilon H^1_i\right) &= (1-\varepsilon)\log\left(\frac{1-\varepsilon}{\varepsilon}\right) + \varepsilon \log\left(\frac{\varepsilon}{1-\varepsilon}\right) \\
        &= (1-2\varepsilon)\log\left(\frac{1-\varepsilon}{\varepsilon}\right).
    \end{align*}
    Since both $P_0$ and $P_1$ are product measures, it holds that
    \begin{align*}
         \mathrm{KL}(P_0||P_1) &= \sum_{i=L+1}^{L+\Delta}\mathrm{KL}((1-\varepsilon) F^0_i+\varepsilon H^0_i||(1-\varepsilon) F^1_i+\varepsilon H^1_i). \\
         & = \Delta(1-2\varepsilon)\log\left(\frac{1-\varepsilon}{\varepsilon}\right)
    \end{align*}
    Hence, we have 
     \begin{equation}\label{kllower2}
        \inf_{\hat{\eta}} \sup_{P\in \mathcal{P}} \mathbb{E}_P(d_H(\hat{\bm{\eta}},\,\bm{\eta}(P))) \geq \frac{\Delta}{2}\exp\left(-\Delta(1-2\varepsilon)\log\left(\frac{1-\varepsilon}{\varepsilon}\right)\right).
    \end{equation}
    Next, set $\Delta = \min\left\{\frac{1}{(1-2\varepsilon)\log\left(\frac{1-\varepsilon}{\varepsilon}\right)}, n-1-L\right\}$. Using our assumption that 
    \[
    (1-2\varepsilon)\log\left(\frac{1-\varepsilon}{\varepsilon}\right)L \geq \zeta_n.
    \]
    where $\zeta_n$ is a diverging sequence, and
    \[
    \zeta_n > \frac{n/2}{n-1-L} \geq \frac{L}{n-1-L}.
    \]
    for $n$ large enough, we must have 
    \[
    (1-2\varepsilon)\log\left(\frac{1-\varepsilon}{\varepsilon}\right)L > \frac{L}{n-1-L}
    \]
    for $n$ large enough. 
    Therefore, we have $\Delta = \frac{1}{(1-2\varepsilon)\log\left(\frac{1-\varepsilon}{\varepsilon}\right)}$ for $n$ large enough, and the claimed bound follows from (\ref{kllower2}).
      \end{proof}
      
\section{Proofs of the results in Section \ref{section3}}\label{proofsection3}

In this section, we prove Theorem \ref{scp2}, utilising Proposition \ref{rume} and ideas in \cite{niu2012screening}. Corollary \ref{nsp} follows straightforwardly from the proof of Theorem \ref{scp2}. Lastly, Proposition \ref{prop2} considers the range of $h$ that satisfies the assumptions in Theorem \ref{scp2}.

\begin{proof}[Proof of Theorem \ref{scp2}]\label{pthm1}
Denote all points that are  more than $2h$ away from any true change point by $\mathcal{F}$, i.e. $\mathcal{F} = \{x:|x-\eta_k|>2h, \,\forall k=1,\dotsc,K\}$. In the first step, we show that under assumptions (i)-(iii), it holds that for $\forall x \in \mathcal{F}$, 
\begin{equation}\label{thm21}
    \left|\left(\text{RUME}\left(\{Y_i\}_{i=x+1}^{x+2h}\right) - \mathbb{E}[F_{x}]\right) - \left(\text{RUME}\left(\{Y_i\}_{i=x-2h+1}^{x}\right)-\mathbb{E}[F_x]\right)\right| \leq \lambda
\end{equation}
with probability at least $1-10\delta$. Note that $x \in \mathcal{F}$ is equivalent to say there is no change point in the interval $[x-2h, \,x+2h]$. Therefore, we have $Y_{x-2h+1},\dotsc,Y_{x+2h}$ are independent random variables with distribution 
\[
(1-\varepsilon_i)F_x+\varepsilon_i H_i,
\]
for $i = x-2h+1,\dotsc,x+2h$.
Without loss of generality, we can assume $\eta_k < x < \eta_{k+1}$, then $\mathbb{E}[F_x] = f_{\eta_{k+1}}$. Under the assumption (ii), we can apply Proposition \ref{rume} and a union bound to get for sufficiently large $C_\lambda$ and the choice of $\lambda = C_\lambda \sigma \sqrt{\varepsilon'}$
\[
\left|\left(\text{RUME}\left(\{Y_i\}_{i=x+1}^{x+2h}\right) - f_{\eta_{k+1}}\right)\right| \leq \frac{\lambda}{2} \quad \text{and}  \quad \left|\left(\text{RUME}\left(\{Y_i\}_{i=x-2h+1}^{x}\right)-f_{\eta_{k+1}} \right)\right|\leq \frac{\lambda}{2}
\]
with probability at least $1-10\delta$. Consequently, equation (\ref{thm21}) follows from the triangle inequality. 

Next, we use similar arguments to show that with probability at least $1-10\delta$
\begin{equation}\label{thm22}
 \left|\left(\text{RUME}\left(\{Y_i\}_{i=\eta_k+1}^{\eta_k+2h}\right) \right) - \left(\text{RUME}\left(\{Y_i\}_{i=\eta_k-2h+1}^{\eta_k}\right)\right) \right| > \lambda.
\end{equation}
for $\forall k = 1,2,\dotsc,K$. Note that the assumption $L>8h>4h$ guarantees that $Y_{\eta_k+1},\dotsc, Y_{\eta_k+2h}$ are independent random variables with distribution $(1-\varepsilon)F_{\eta_{k+1}}+\varepsilon H_{i}$ for $i = \eta_k+1, \dotsc, \eta_k+2h$ where $\mathbb{E}[F_{\eta_{k+1}}] = f_{\eta_{k+1}}$ and $Y_{\eta_k-2h+1}, \dotsc,Y_{\eta_k}$ are independent random variables with distribution $(1-\varepsilon_j)F_{\eta_{k}}+\varepsilon_j H_{j}$ for $j = \eta_k-2h+1, \dotsc, \eta_k$, where $\mathbb{E}[F_{\eta_{k}}] = f_{\eta_k}$. We take square of the left hand side of equation (\ref{thm22}) and rewrite it as 
\begin{align*}
     &\left|\left(\text{RUME}\left(\{Y_i\}_{i=\eta_k+1}^{\eta_k+2h}\right) - f_{\eta_{k+1}}\right) - \left(\text{RUME}\left(\{Y_i\}_{i=\eta_k-2h+1}^{\eta_{k}}\right)-f_{\eta_{k}}\right) + (f_{\eta_k+1}-f_{\eta_{k}})\right|^2 \\
    & \geq \kappa^2/2-\left|\left(\text{RUME}\left(\{Y_i\}_{i=\eta_k+1}^{\eta_k+2h}\right) - f_{\eta_{k+1}}\right) - \left(\text{RUME}\left(\{Y_i\}_{i=\eta_k-2h+1}^{\eta_{k}}\right)-f_{\eta_{k}}\right)\right|^2,
\end{align*}
where the inequality follows from the observation that $(x+y)^2 \geq x^2/2 - y^2$ for any $x,y \in \mathbb{R}$. Using Proposition \ref{rume} and triangle inequality, we have 
\[
\left|\left(\text{RUME}\left(\{Y_i\}_{i=\eta_k+1}^{\eta_k+2h}\right) - f_{\eta_{k+1}}\right) - \left(\text{RUME}\left(\{Y_i\}_{i=\eta_k-2h+1}^{\eta_{k}}\right)-f_{\eta_{k}}\right)\right|^2 \leq \lambda^2
\]
with probability at least $1-10\delta$. Together with the assumption that $\kappa>2\lambda$, we have 
\[
 \left|\left(\text{RUME}\left(\{Y_i\}_{i=\eta_k+1}^{\eta_k+2h}\right) \right) - \left(\text{RUME}\left(\{Y_i\}_{i=\eta_k-2h+1}^{\eta_k}\right)\right) \right|^2 \geq \kappa^2/2-\lambda^2 > \lambda^2.
\]

In the second step, we consider the following events 
\begin{align*}
       B_x &= \{|D_h(x)|<\lambda\}  \\
A_{\eta_k}  &= \{|D_h{(\eta_k)}|>\lambda\} \\
\mathcal{E}_n &= \left(\cap_{k=1}^K A_{\eta_k}\right) \cap \left(\cap_{\substack{x\in \mathcal{F}}}B_x\right)
\end{align*}
and argue that on the event $\mathcal{E}_n$, we have 
\[
\hat{K}=K \quad \text{and} \quad \max_{k=1,\dotsc,\hat{K}}|\hat{\eta}_k-\eta_k|\leq 2h.
\]

Note that it is sufficient to show that on the event $\mathcal{E}_n$ it holds that i) for any estimated change point $ \hat{\eta}_k, \;k = 1,2\dotsc, \hat{K}$, there is a unique true change point $\eta_k$ lying in the interval $(\hat{\eta}_k-2h,\,\hat{\eta}_k+2h) $ and ii) for each true change point $\eta_k$, $k = 1,2,\dotsc,K$, there is a unique estimated change point located in the interval $(\eta_k-2h,\,\eta_k+2h)$. 

For i), we notice that $\hat{\eta}_k \in \mathcal{F}^c$ for all $k = 1,2,\dotsc,\hat{K}$, according to the definition of event $\mathcal{E}_n$. Therefore, there is at least one true change point in the interval $(\hat{\eta}_k-2h,\,\hat{\eta}_k+2h)$. Using the assumption $L>8h>4h$, we see there is at most one true change point in $(\hat{\eta}_k-2h,\,\hat{\eta}_k+2h)$. Therefore i) holds. For ii), using the assumption $L>8h$, we know every point in the intervals $(\eta_k+2h,\eta_k+6h)$ and $(\eta_k-2h,\eta_k-6h)$ belong to $\mathcal{F}$. This means $|D_h(x)|< \lambda$ for all x in the aforementioned two intervals. Therefore the $4h$ local maximizers of $|D_h(x)|$ for $x \in (\eta_k -2h,\eta_k+2h)$ correspond to the unique local maximizer $\eta^*$ of $|D_h(x)|$ for $x \in (\eta_k -2h,\eta_k+2h)$, and we have $|D_h(\eta^*)|\geq |D_h(\eta_k)|>\lambda$.  

In the last step, we show that $\mathbb{P}(\mathcal{E}_n^c) \rightarrow$ 0 using (\ref{thm21}) and (\ref{thm22}) from step 1. Note that using union bound, we have
\begin{equation} \label{complement}
    \mathbb{P}(\mathcal{E}_n^c) \leq \mathbb{P}(\cup_{k=1}^K A_{\eta_k}^c)+\mathbb{P}(\cup_{\substack{x\in \mathcal{F}}}B_x^c) \leq n\max_{k}\mathbb{P}(A_{\eta_k}^c)+n\max_{x\in \mathcal{F}}\mathbb{P}(B_x^c).
\end{equation}
Using (\ref{thm22}), we have
\[
\max_{k}\mathbb{P}(A_{\eta_k}^c) = \max_{k}\mathbb{P}\left(\left|\left(\text{RUME}\left(\{Y_i\}_{i=\eta_k+1}^{\eta_k+2h}\right) \right) - \left(\text{RUME}\left(\{Y_i\}_{i=\eta_k-2h+1}^{\eta_k}\right)\right) \right| <\lambda. \right)\leq 10\delta
\]
Using (\ref{thm21}), we have 
\begin{align*}
\max_{x\in \mathcal{F}}\mathbb{P}(B_x^c) = \mP\left( \left|\left(\text{RUME}\left(\{Y_i\}_{i=x+1}^{x+2h}\right) - \mathbb{E}[Y_{x+1}]\right) - \left(\text{RUME}\left(\{Y_i\}_{i=x-2h+1}^{x}\right)-\mathbb{E}[Y_x] \right)\right| > \lambda \right) \\
\leq 10 \delta.
\end{align*}

Combining the upper bounds for $\mathbb{P}(A_{\eta_k}^c)$ and $\mathbb{P}(B_{x}^c)$, we can conclude that
\[
\mP(\mathcal{E}_n^c) \leq 20n\delta = 20n^{1-C'} \rightarrow 0 
\]
under the choice that $\delta = n^{-C'}$ for some constant $C'>1$.
\end{proof}

\begin{proof}[Proof of Corollary \ref{nsp}]
Since $\kappa = 0$, we have $Y_1, \dotsc, Y_n$ are independent random variables with distribution 
\[
(1-\varepsilon_i)F_i+\varepsilon_i H_i
\]
for $i = 1,\dotsc,n$, where $\mathbb{E}[F_i] = f_1$. Using (\ref{thm21}) from the Proof of Theorem \ref{scp2}, we have for $x = 2h,\dotsc, n-2h$, it holds that $D_h(x) < \lambda$ with probability at least $1-10\delta$. Using the notation from the Proof of Theorem \ref{scp2}, we have on the event $\cap_{x = 2h}^{n-2h}B_x$, it holds that $\hat{K} = 0$. It follows from a union bound that 
\[
\mP\left(\bigcup_{x = 2h}^{n-2h}B_x^c\right) \leq 10n\delta,
\]
Therefore, the claim follows by choosing $\delta = n^{-C'}$ for some constant $C'>1$.
\end{proof}

\begin{prop}\label{prop2} Under the same notation as in Theorem \ref{scp2}, the following choices of $h$ can guarantee that assumptions (i)-(iii) holds
\begin{gather*}
    h>\max\left\{10, \,4C_\lambda^2\frac{\sigma^2}{\kappa^2}\right\}C'\log(n), \qquad \text{if}\quad \varepsilon < 0.1, \\
    h>  h(\varepsilon)C'\log(n) \qquad \text{if} \quad 0.1<\varepsilon<\frac{1}{4}\min{\left\{1, \,\frac{\kappa^2}{\sigma^2C_\lambda^2}\right\}},
\end{gather*}
where $$h(\theta) = \frac{1}{0.5-\sqrt{2\theta(1-2\theta)}}.$$
\end{prop}

\begin{proof}[Proof of Proposition \ref{prop2}]
We consider two separate cases:
\begin{enumerate}
    \item When $\varepsilon' = \varepsilon$, which is equivalent to 
    \[
    h >\frac{C'\log(n)}{\varepsilon},
    \]
    assumption (ii) can be simplified as 
    \begin{gather*}
        2\varepsilon+2\sqrt{\varepsilon\frac{C'\log(n)}{h}}+\frac{C'\log(n)}{h} < \frac{1}{2}, \\
        \sqrt{\varepsilon}+\sqrt{\frac{C'\log(n)}{h}}<\sqrt{\frac{1}{2}-\varepsilon} \\
        \mbox{and} \quad h> \left(\frac{\sqrt{C'\log(n)}}{\sqrt{1/2-\varepsilon}-\sqrt{\varepsilon}}\right)^2 = h(\varepsilon) C'\log(n).
    \end{gather*}
    Note that $h(\varepsilon) \geq 1/\varepsilon $ for $0.1<\varepsilon<0.25$. Combining with assumption (i), we have if $0.1<\varepsilon<\frac{1}{4}\min{\left\{1,\frac{\kappa^2}{\sigma^2C_\lambda^2}\right\}}$, then $h>h(\varepsilon)\log(n)$ can satisfy the assumptions for Theorem \ref{scp2}.
    \item When $\varepsilon' = C'\log(n)/h$, which is equivalent to 
    \[
    h < \frac{C'\log(n)}{\varepsilon},
    \]
    the assumption (i) requires 
    \[
    h > 4C_\lambda^2 C' \frac{\sigma^2}{\kappa^2}\log(n),
    \]
    and assumption (ii) requires 
    \[
    h> 10C' \log(n).
    \]
    Combining these three conditions, we have if $h$ satisfies  
    \begin{equation}\label{3.1}
        \max\left\{10, \,4C_\lambda^2\frac{\sigma^2}{\kappa^2}\right\}C'\log(n)< h < \frac{C'\log(n)}{\varepsilon},
    \end{equation}
    then assumptions in Theorem \ref{scp2} are satisfied. To ensure that (\ref{3.1}) is not an empty set, it is sufficient to require $\varepsilon < 0.1$.
\end{enumerate}
\end{proof}

\section{Further details of the numerical results in Section \ref{numerical}}\label{detailnum}

\subsection{Tuning parameter selection in simulations}\label{tuningsele}

 For \textsc{Biweight}~\cite{fearnhead2019}, we choose (using their notation) the default value as $K = 3\sigma$ and $\beta = 2\sigma^2\log(n) = 17\sigma^2$, where $\sigma$ is the standard deviation of the uncontaminated distribution $F_i$. We denote this choice of tuning paramter as \textsc{Biweight}(2). It is noted in \cite{fearnhead2019} that such choice is not guaranteed to ensure consistency. Therefore, we also consider a stronger penalty value $\beta = 5\sigma^2\log(n) = 42.6\sigma^2$ in the simulation (denoted as \textsc{Biweight(5)}). For the \textsc{R\_cusum} procedure, we combine it with the wild binary segmentation framework \cite{fryzlewicz2014wild} for multiple change point scenarios. We generate 500 random intervals with threshold set to $5\sigma^2\log(n) = 42.6\sigma^2$. 

For the robust $U$-statistics test, we only consider the univariate version of their proposed statistic $T_n$ here and the distribution of $T_n$ under the null (no change points) is approximated by the Gaussian multiplier bootstrap. Combining with the backward detection algorithm (c.f.~Algorithm 1 in~\cite{yu2019robust}), this robust test can be used to perform multiple change point detection. 

For a\textsc{arc}, the value $\varepsilon$ is chosen via the tournament procedure proposed by \cite{chen2016general}.  We include the details here for completeness.

We use the first $T = 300$ of each simulated data set as the training set and 
obtain a sequence of estimates $\{\hat{\theta}_1,\dotsc,\hat{\theta}_m\}$ returned by the RUME with input $\{\varepsilon_1,\dotsc,\varepsilon_m\}$ which is an equally spaced set from $0$ to $0.25$ with size $m = 201$. Consider the following pairwise test function:
\begin{align*}
    \phi_{jk} &= \mathds{1}\left\{\Bigg|\frac{1}{T}\sum_{i=1}^T \mathds{1}\left\{p_{\hat{\theta}_j}(Y_i)>p_{\hat{\theta}_k}(Y_i)\right\} - P_{\hat{\theta}_j}\left(Y > \frac{\hat{\theta}_j + \hat{\theta}_k}{2}\right) \Bigg| > \right.\\
    & \qquad \quad \left.\Bigg| \frac{1}{T}\sum_{i=1}^T \mathds{1}\left\{p_{\hat{\theta}_j}(Y_i)>p_{\hat{\theta}_k}(Y_i)\right\} - P_{\hat{\theta}_k}\left(Y > \frac{\hat{\theta}_j + \hat{\theta}_k}{2}\right) \Bigg|\right\},
\end{align*}
where $p_{\hat{\theta}_j}$ is the probability density of $P_{\hat{\theta}_j} = \mathcal{N}(\hat{\theta}_j,\,\sigma^2)$. When $\phi_{jk} = 1$, then $\hat{\theta}_k$ is favoured over $\hat{\theta}_j$, and when $\phi_{jk} = 0$, then $\hat{\theta}_k$ is favoured over $\hat{\theta}_j$. We select $\varepsilon_{j^*}$ that corresponds to the estimate $\hat{\theta}_{j^*}$ where 
\begin{equation}\label{tuning}
    j^* = \argmin_{j = 1,\dotsc,m} \sum_{k \neq j}\phi_{jk}.
\end{equation}

It is shown in \cite{chen2016general} that the above procedure would pick a $j^*$ such that $P_{\hat{\theta}_{j^*}}$ is close to $\mathcal{N}(\theta,\,\sigma^2)$ in total variance metric provided that the training data are independent samples from the Huber's $\varepsilon$-contamination model (\ref{huber0}) with $F = \mathcal{N}(\theta,\,\sigma^2)$ and fixed contamination distribution $H$. Note that we consider here specifically the case when $F$, the uncontaminated distribution, is Gaussian in the Huber $\varepsilon$-contamination model (\ref{huber0}). Other classes of distributions of $F$ can also be considered by using the corresponding density function. 

\subsection{Further simulation results}\label{furthurnumerical}

\subsubsection{Adversarial settings (i) and (ii)}
In this subsection, we provide the complete simulation results for the two adversarial attack settings that are considered in our paper. Tables~\ref{spurious} and \ref{table2} correspond to scenario (i) and (ii), respectively. The last column in Table \ref{spurious} shows the mean error in estimating the number of change points and if $\hat{K} = 2\Delta-1$, it means the algorithm detects the number of spurious change points created by the adversarial noise. The last column in Table \ref{table2} shows the median (rescaled) Hausdorff distance and if $\hat{K} = 0$, it means the algorithm cannot detect the true change points in the presence of contamination.

\begin{longtable}{c c c c c c c}
            \caption{Estimated number of change points for various competing methods over 100 simulations when the adversarial noise tries to create spurious change points. The number of change points in terms of $f_i$ is $K = 0$ while the number of change points in terms of $\mathbb{E}[Y_i]$ is $2\Delta-1$. Bold: methods with the lowest mean error for estimating the number of change point $K$.}\label{spurious} \\

			\toprule
			
			&&&&\multicolumn{2}{c}{Number of detected change points} \\
			$\varepsilon$ & $\Delta$ & $\sigma$ & Method & $ \hat{K} = K$ & $\hat{K} = 2\Delta-1$ &  $(\hat{K}-K)/100$  \\
			
			\midrule
			
			0 & 1 & 1 & \textsc{pelt} & 100 & 0 & \textbf{\textbf{0.00}} \\
			&&& \textsc{arc} & 100 & 0 & \textbf{\textbf{0.00}} \\
			& & & a\textsc{arc} & 77 & 0 & 0.30 \\
			 &  & & \textsc{Biweight}(2) & 100 & 0 & \textbf{\textbf{0.00}} \\
			 &  & & \textsc{Biweight}(5) & 100 & 0 & \textbf{\textbf{0.00}} \\
			 &  &  & \textsc{R\_cusum} & 100 & 0 & \textbf{\textbf{0.00}} \\
			\vspace{1em}
			 & &  & \textsc{R\_UStat} & 100 & 0 & \textbf{\textbf{0.00}} \\
			
			0.05&1& 1& \textsc{pelt} & 0 & 100 & 1.00 \\
			& & &  \textsc{arc} & 85 & 15 & 0.15  \\ 
				& & & a\textsc{arc} & 72 & 10 & 0.56 \\
			&& & \textsc{Biweight}(2) & 0 & 72 & 1.50\\
			&& & \textsc{Biweight}(5) & 1 & 99 & 0.99 \\
			&& & \textsc{R\_cusum} &  97 & 3 & \textbf{0.03}  \\
			&& & \textsc{R\_UStat} &  0 &  88 & 1.15  \\
			
			0.05&1& 5& \textsc{pelt} & 95 & 5 & 0.05 \\
			& & &  \textsc{arc} & 83 & 16 & 0.18  \\ 
				& & & a\textsc{arc} & 85 & 14 & 0.16 \\
			&& & \textsc{Biweight}(2) & 95 & 5 & 0.05\\
			&& & \textsc{Biweight}(5) & 100 & 0 & \textbf{\textbf{0.00}} \\
			&& & \textsc{R\_cusum} &  100 & 0 & \textbf{\textbf{0.00}}  \\
			&& & \textsc{R\_UStat} &  35 &  55 & 0.79  \\
			
			0.05&1& 20& \textsc{pelt} & 100 & 0 & \textbf{\textbf{0.00}} \\
			& & &  \textsc{arc} & 98 & 2 & 0.02  \\ 
				& & & a\textsc{arc} & 97 & 3 & 0.03 \\
			&& & \textsc{Biweight}(2) & 98 & 2 & 0.02\\
			&& & \textsc{Biweight}(5) & 100 & 0 & \textbf{\textbf{0.00}} \\
			&& & \textsc{R\_cusum} &  100 & 0 & \textbf{\textbf{0.00}}  \\
			&& & \textsc{R\_UStat} &  76 &  15 & 0.34  \\
			
			0.05&5&1 & \textsc{pelt} & 0 & 98 & 0.02 \\
			& & &  \textsc{arc} & 89 & 0 & 0.15  \\ 
				& & & a\textsc{arc} & 77 & 0 & 0.39 \\
			&& & \textsc{Biweight}(2) & 9 & 4 & 4.75\\
			&& & \textsc{Biweight}(5) & 96 & 0 & 0.04 \\
			&& & \textsc{R\_cusum} &  100 & 0 & \textbf{\textbf{0.00}} \\
			&& & \textsc{R\_UStat} &  50 &  6 & 2.14  \\
			
			0.05&5&5 & \textsc{pelt} & 100 & 0 & \textbf{\textbf{0.00}} \\
			& & &  \textsc{arc} & 73 & 0 & 0.36  \\ 
				& & & a\textsc{arc} & 67 & 0 & 0.53 \\
			&& & \textsc{Biweight}(2) & 100 & 0 & \textbf{\textbf{0.00}}\\
			&& & \textsc{Biweight}(5) & 100 & 0 & \textbf{\textbf{0.00}} \\
			&& & \textsc{R\_cusum} &  100 & 0 & \textbf{\textbf{0.00}}  \\
			&& & \textsc{R\_UStat} &  87 &  0 & 0.24  \\
			
			0.05&5&20 & \textsc{pelt} & 100 & 0 & \textbf{\textbf{0.00}} \\
			&&& \textsc{arc} & 95 & 0 & 0.05\\
				& & & a\textsc{arc} & 89 & 0 & 0.12 \\
			&& & \textsc{Biweight}(2) & 100 & 0 & \textbf{\textbf{0.00}}\\
			&& & \textsc{Biweight}(5) & 100 & 0 & \textbf{\textbf{0.00}}\\
			&& & \textsc{R\_cusum} &  100 & 0 & \textbf{\textbf{0.00}}  \\
			\vspace{1em}
			&& & \textsc{R\_UStat} &  95 & 0 &0.12      \\
			
			0.1&1&1 & \textsc{pelt} & 0 & 100 & 1.00 \\
			& & &  \textsc{arc} & 83 & 12 & $\bm{0.24}$   \\ 
				& & & a\textsc{arc} & 84 & 4 & 0.47 \\
			&& & \textsc{Biweight}(2) & 0 & 37 & 2.96\\
			&& & \textsc{Biweight}(5) & 0&100&1.00 \\
			&& & \textsc{R\_cusum} &  0 & 100 & 1.00  \\
			&& & \textsc{R\_UStat} &  0 &  79 & 1.29  \\
			
			0.1&1&5 & \textsc{pelt} & 98 & 2 & \textbf{0.02}\\
			&&&  \textsc{arc} &  72& 25 &0.31  \\ 
				& & & a\textsc{arc} & 80 & 20 & 0.20 \\
			&& & \textsc{Biweight}(2) &   27 &70 &0.76 \\
			&&& \textsc{Biweight}(5) & 94 & 6 & 0.06 \\
			&& & \textsc{R\_cusum} & 97&3&0.03   \\
			&& & \textsc{R\_UStat} & 0 & 82 & 1.23   \\
			
			0.1&1&20 & \textsc{pelt} & 100 & 0 & \textbf{\textbf{0.00}} \\
			&&&  \textsc{arc} & 96 & 3 & 0.05 \\
				& & & a\textsc{arc} & 99 & 1 & 0.01 \\
			&& & \textsc{Biweight}(2) & 100 & 0 & \textbf{\textbf{0.00}} \\
			&& & Biwerght(5) & 100 & 0& \textbf{\textbf{0.00}} \\
			&& & \textsc{R\_cusum} &  100&0& \textbf{\textbf{0.00}}\\
			&& & \textsc{R\_UStat} &  51&41&0.57  \\
			
			0.1&2&1 & \textsc{pelt} & 0 & 99 & 3.01\\
			&&&  \textsc{arc} & 86 & 0 & \textbf{\textbf{0.17}} \\
				& & & a\textsc{arc} & 80 & 3 & 0.61 \\
			&& & \textsc{Biweight}(2) & 0 & 28 & 5.23 \\
			&& & Biwerght(5) & 0 & 100 & 3 \\
			&& & \textsc{R\_cusum} &  30 & 20 & 1.43\\
			&& & \textsc{R\_UStat} &  1 & 28 & 2.63  \\
			
			0.1&5&1 & \textsc{pelt} & 2 & 80 & 8.07 \\
			&&& \textsc{arc} & 85 & 0 & 0.19\\ 
			& & & a\textsc{arc} & 85 &0 & 0.54 \\
			&& & \textsc{Biweight}(2) & 0 & 18 & 11.67 \\
			&&&\textsc{Biweight}(5) & 6 & 10 & 4.42 \\
			&& & \textsc{R\_cusum} & 99 & 0 & \textbf{0.02}           \\
			&& & \textsc{R\_UStat} & 0 & 61 & 9.40\\
			
			0.1&5&5 & \textsc{pelt} & 100 & 0 & \textbf{\textbf{0.00}} \\
			&&& \textsc{arc} & 74 & 0 & 0.29\\ 
			& & & a\textsc{arc} & 52 &0 & 1.07 \\
			&& & \textsc{Biweight}(2) & 99 & 0 &0.01 \\
			&&&\textsc{Biweight}(5) & 100 & 0 & \textbf{\textbf{0.00}}\\
			&& & \textsc{R\_cusum} & 100 & 0 & \textbf{\textbf{0.00}} \\
			&& & \textsc{R\_UStat} & 58 & 3 & 1.50\\
			
			0.1 & 5 & 20 & \textsc{pelt} & 100 & 0 & \textbf{\textbf{0.00}} \\
			&&& \textsc{arc} & 90 & 0 & 0.14\\
				& & & a\textsc{arc} & 97 & 0 & 0.03 \\
				&  &  & \textsc{Biweight}(2) & 100 & 0 & \textbf{\textbf{0.00}}\\
				&  &  & \textsc{Biweight}(5) & 100 & 0 & \textbf{\textbf{0.00}}\\
			&& & \textsc{R\_cusum} & 85&0& 0.29           \\
			\vspace{1em}
		
			&& & \textsc{R\_UStat} & 63&0&0.61 \\
			
			0.2&1&1 & \textsc{pelt} & 0 & 100 & 1.00\\
			&&&  \textsc{arc} &  62& 28 & \textbf{0.50}  \\ 
				& & & a\textsc{arc} & 56 & 21 & 0.77 \\
			&& & \textsc{Biweight}(2) &   0 &3 &9.48 \\
			&&& \textsc{Biweight}(5) & 0 & 100 & 1.00 \\
			&& & \textsc{R\_cusum} & 0&100&1.00   \\
			&& & \textsc{R\_UStat} & 0 & 82 & 1.25   \\
			
			0.2&1&5 & \textsc{pelt} & 1 & 99 & 0.99\\
			&&& \textsc{arc} &  75& 25 &0.25  \\ 
			& & & a\textsc{arc} & 99 & 1 & \textbf{0.01} \\
	&& & \textsc{Biweight}(2) &   0 &100 &1.00 \\
			&&& \textsc{Biweight}(5) & 1 & 99 & 0.99 \\
			&& & \textsc{R\_cusum} & 0&100&1.00   \\
			&& & \textsc{R\_UStat} & 0 & 79 & 1.28   \\
			
		0.2&1&20 & \textsc{pelt} & 100 & 0 & \textbf{\textbf{0.00}}\\
			&&&  \textsc{arc} &  100 & 0 & \textbf{\textbf{0.00}}  \\ 
				& & & a\textsc{arc} & 100 & 0 & \textbf{\textbf{0.00}} \\
			&& & \textsc{Biweight}(2) &   97 &3 &0.03 \\
			&&& \textsc{Biweight}(5) & 100 & 0 & \textbf{\textbf{0.00}} \\
			&& & \textsc{R\_cusum} & 100 &0& \textbf{\textbf{0.00}}   \\
			&& & \textsc{R\_UStat} & 0 & 80 & 1.27   \\
			
			0.2&2&1 & \textsc{pelt} & 0 & 100 & 3.00 \\
			&&&  \textsc{arc} &  57& 2 & \textbf{0.53}  \\ 
				& & & a\textsc{arc} & 59 & 7 & 0.84 \\
			&& & \textsc{Biweight}(2) &   0 &3 &11.64 \\
			&&& \textsc{Biweight}(5) & 0 & 100 & 3.00 \\
			&& & \textsc{R\_cusum} & 0&100& 3.00   \\
			&& & \textsc{R\_UStat} & 0 & 33 & 3.01   \\
			
			0.2&5&1 & \textsc{pelt} & 0 & 100 & 9.00 \\
			&&&  \textsc{arc} &  74& 0 & \textbf{0.36}  \\ 
				& & & a\textsc{arc} & 71 & 0 & 0.48 \\
			&& & \textsc{Biweight}(2) &   0 &5 &18.48 \\
			&&& \textsc{Biweight}(5) & 0 & 99 & 9.01 \\
			&& & \textsc{R\_cusum} & 0&71&8.36   \\
			&& & \textsc{R\_UStat} & 0 & 70 & 9.52   \\
			
			0.2&5&5 & \textsc{pelt} & 100 & 0 & \textbf{\textbf{0.00}}\\
			&&&  \textsc{arc} &  67 & 0 &1.56  \\ 
				& & & a\textsc{arc} & 93 & 0 & 0.12 \\
			&& & \textsc{Biweight}(2) &   69 & 0 &0.54 \\
			&&& \textsc{Biweight}(5) & 100 & 0 & \textbf{\textbf{0.00}} \\
			&& & \textsc{R\_cusum} & 100 & 0 & \textbf{\textbf{0.00}}   \\
			&& & \textsc{R\_UStat} & 0 & 51 & 9.26   \\
			
			0.2&5&20 & \textsc{pelt} & 100 & 0 & \textbf{\textbf{0.00}} \\
			&&&  \textsc{arc} & 100 & 0 & \textbf{\textbf{0.00}}  \\ 
				& & & a\textsc{arc} & 100 & 0 & \textbf{\textbf{0.00}} \\
			&& & \textsc{Biweight}(2) &   99 &0 &0.01 \\
			&&& \textsc{Biweight}(5) & 100 & 0 & \textbf{\textbf{0.00}} \\
			&& & \textsc{R\_cusum} & 100 & 0 & \textbf{\textbf{0.00}}   \\
			&& & \textsc{R\_UStat} & 73 & 1 & 0.99   \\
			
			 \bottomrule

\end{longtable}

\begin{longtable}{c c c c c c c H c H H c H}
\caption{Estimated number of change points for various competing methods over 100 simulations when the adversarial noise tries to hide change points. The number of change points in terms of $f_i$ is $K = 2\Delta-1$ while there is no change points in terms of $\mathbb{E}[Y_i]$. Also the median Hausdorff distance divided by sample size and the number of repetitions that the Hausdorff distance is less than 2$h$. Bold: methods with the smallest and second smallest mean error for estimating the number of change point $K$ and median Hausdorff distance divided by sample size.} 
\label{table2}
			\\ \toprule
			&&&&\multicolumn{4}{c}{Number of detected change points} &\multicolumn{5}{c}{$d_H(\hat{\bm{\eta}},\,\bm{\eta})$ }\\
			$\varepsilon$ & $\Delta$ & $\kappa$ & Method & $\hat{K} = 0$ & $\hat{K} = K$ & $|\hat{K}-K|/100$ & Otherwise & $\leq 2h$ & $\leq h$ & $>2h$ & median & mean \\
			\midrule 
			0&1&0.6 & \textsc{pelt} & 0 & 100& \textbf{0.00} & nothing & \textbf{100} &0 & nothing & \textbf{0.00} & \textbf{0.00} \\
			&&& \textsc{Biweight}(2) &0 & 100 & \textbf{0.00} &0 & \textbf{100} & 100 & 0 & \textbf{0.00}  \\
			&& &  \textsc{Biweight}(5) & 0 & 100 & \textbf{0.00} & 0 & \textbf{100} & 100 & 0 & \textbf{0.00}  \\ 
			&& & \textsc{arc} & 0 & 100 & \textbf{0.00} & 0 & \textbf{100}& 100 & 0 & \textbf{0.00}  \\
				&& & a\textsc{arc} & 9 & 78 & 0.22 & \textbf{0.00} & 75 & 85 & 13 & \textbf{0.01} \\
			& & & \textsc{R\_cusum} &   0 & 98 & \textbf{0.00} & 0 &\textbf{99} & 99& 1 & \textbf{0.00}\\
				& & & \textsc{R\_UStat} & 0 & 82 & 0.30& 12 & 82 & 0 & 93 &\textbf{0.00}\\ 
				
				0&2&0.6 & \textsc{pelt} & 0 & 100& \textbf{0.00} & nothing & \textbf{100} &0 & nothing & \textbf{0.00} & 5000 \\
				
			&&	& \textsc{Biweight}(2) & 0 & 100 & \textbf{0.00} &0 & \textbf{100} &100 & 0 & \textbf{0.00} \\
			&& &  \textsc{Biweight}(5) & 0 & 100 & \textbf{0.00} &0 & \textbf{100} &100 & 0 & \textbf{0.00} \\ 
			&& & \textsc{arc} & 0 & 100 & \textbf{0.00} & \textbf{0.00} & \textbf{100} & 100 & 0 & \textbf{0.00}  \\
			&& & a\textsc{arc} & 0 & 88 & 0.12 & \textbf{0.00} & \textbf{87} & 85 & 13 & \textbf{0.01} \\
			& & & \textsc{R\_cusum} &   0 & 99 & \textbf{0.01} & 0 &\textbf{100} & 100& 0 & \textbf{0.00}\\
				
			\vspace{1em}
			& & & \textsc{R\_UStat} & 0 & 70 & 0.44 & 70 & 70 & 0 & 93 & \textbf{0.00}\\ 
			
			0.1 & 1 & 0.6 & \textsc{pelt} & 100 & 0& 1.00 & nothing & 0 &0 & nothing & 1.00 & 5000 \\
			& && \textsc{Biweight}(2) & 39 & 2 & 1.57 &nothing& 0 & 0 &nothing& 0.58 & \textbf{2420.5} \\ 
			& & & \textsc{Biweight}(5)  & 100 & 0& 1.00 & Nothing & 0 &0 & nothing & 1.00 & 5000 \\
			 & & &  \textsc{arc} & 46 & 36 & 0.70 & NOTHING & 33 & 29 & Nothing & 0.31 & 2572.66 \\ 
			 && & a\textsc{arc} & 46 & 32 & 0.96 & Nothing & 32 & \textbf{32} & Nothing & \textbf{0.33} & 2621.66 \\
			& & & \textsc{R\_cusum} &   60 & 40 & \textbf{0.60} & Nothing & \textbf{40} & \textbf{33} & Nothing & 0.50 & \textbf{1529}\\
			& & & \textsc{R\_UStat} & 0 & 91 & \textbf{0.12} & Nothing & \textbf{90} & 3 & Nothing & \textbf{0.00} &3665\\ 
			0.1 & 1 & 0.66 & \textsc{pelt} & 100 & 0& 1.00 & nothing & 0 &0 & nothing & 1.00 & 5000 \\
			& && \textsc{Biweight}(2) & 0 & 9 & 4.85 &nothing& 17 & 14 &nothing& 0.37 & 1558.45 \\ 
			& & & \textsc{Biweight}(5)  & 0 & 100 & \textbf{\textbf{0.00}} & Nothing & \textbf{100} & \textbf{97} & nothing & \textbf{\textbf{0.00}} & \textbf{39.98} \\
			 & & &  \textsc{arc} & 7 & 93 & \textbf{0.07} & NOTHING & \textbf{92} & 50 & Nothing & \textbf{0.00} & 2100.64 \\ 
			 && & a\textsc{arc} & 39 & 52 & 0.59 & Nothing & 52 & 51 & Nothing & 0.04 & 2086.47 \\
			& & & \textsc{R\_cusum} &   8 & 92 & 0.08 & Nothing &91 & \textbf{85} & Nothing & \textbf{0.01} & \textbf{247.72}\\
			& & & \textsc{R\_UStat} & 0& 84 & 0.28 & Nothing & 84 & 3 & Nothing & \textbf{0.00} &3127.5\\ 
				0.1 & 1 & 1 & \textsc{pelt} & 70 & 2 & 1.05 & nothing & 1 &0 & nothing & 1.00 & 3983.96 \\
			& && \textsc{Biweight}(2) & 0 & 48 & 1.28 &nothing& 52 & 50 &nothing& \textbf{0.03} & 826.7 \\ 
			& & & \textsc{Biweight}(5)  & 0 & 100 & \textbf{\textbf{0.00}} & Nothing & \textbf{100} & 100 & nothing & \textbf{\textbf{0.00}} & \textbf{4.19} \\
			 & & &  \textsc{arc} & 3 & 87 & 0.15 & NOTHING & 87 & 87 & Nothing & \textbf{\textbf{0.00}} & 306.64 \\ 
			 && & a\textsc{arc} & 2 & 93 & \textbf{0.08} & Nothing & \textbf{93} & 92 & Nothing & \textbf{\textbf{0.00}} & 197.83 \\
			& & & \textsc{R\_cusum} &   0 & 100 & \textbf{\textbf{0.00}} & Nothing & \textbf{100} & \textbf{100} & Nothing & \textbf{\textbf{0.00}} & \textbf{14.11} \\
			& & & \textsc{R\_UStat} & 0 & 82 & 0.25 & Nothing & 82 & 0 & Nothing & \textbf{0.00}&3072.5\\ 
			0.1 & 2 & 0.6  & \textsc{pelt} & 100 & 0& 3.00 & nothing & 0 &0 & nothing & 1.00 & 5000 \\
			& && \textsc{Biweight}(2) & 58 & 2 & 2.39 &nothing& 0 & 0 &nothing& 1.00 & 3807.29 \\ 
			& & & \textsc{Biweight}(5)  & 100 & 0& 3.00 & Nothing & 0 &0 & nothing & 1.00 & 5000 \\
			 & & &  \textsc{arc} & 35 & 32 & \textbf{1.55} & NOTHING & \textbf{32} & \textbf{29} & Nothing & 0.36 & 2342.85 \\ 
			 && & a\textsc{arc} & 29 & 17 & 1.71 & Nothing & 15 & 12 & Nothing & \textbf{0.25} & \textbf{2206.24} \\
			& & & \textsc{R\_cusum} &   99 & 0 & 2.99 & Nothing & 0 & 0 & Nothing & 0.75 & 3738.9\\
			& & & \textsc{R\_UStat} & 15 & 62 & \textbf{0.80} & Nothing & \textbf{54} & 0 & Nothing & \textbf{0.05} & 3190\\ 
			0.1 & 2 & 0.66  & \textsc{pelt} & 100 & 0& 3.00 & nothing & 0 &0 & nothing & 1.00 & 5000 \\
			& && \textsc{Biweight}(2) & 0 & 6 & 5.55 &nothing& 21 & 12 &nothing& 0.12 & \textbf{661.81} \\ 
			& & & \textsc{Biweight}(5)  & 0 & 100 & \textbf{\textbf{0.00}} & Nothing & \textbf{96} & \textbf{81} & nothing & \textbf{0.02} & \textbf{118.61} \\
			 & & &  \textsc{arc} & 0 & 77 & 0.81 & NOTHING & 78 & \textbf{38} & Nothing & \textbf{0.01} & 1687.48 \\ 
			 && & a\textsc{arc} & 14 & 44 & 1.09 & Nothing & 42 & 35 & Nothing & 0.23 & 1489.71 \\
			& & & \textsc{R\_cusum} &   77 & 4 & 2.6 & Nothing & 4 & 4 & Nothing & 0.75 & 3214.88\\
			& & & \textsc{R\_UStat} & 3 & 64 & \textbf{0.62} & Nothing & \textbf{59} & 0 & Nothing & 0.05 & 2700\\ 
				0.1 & 2 & 1 & \textsc{pelt} & 73 & 1 & 2.55 & nothing & 0 &0 & nothing & 1.00 & 4245.65 \\
			& && \textsc{Biweight}(2) & 0 & 57 & 1.47 &nothing& 75 & 75 &nothing& \textbf{\textbf{0.00}} & 225.36 \\ 
			& & & \textsc{Biweight}(5)  & 0 & 100 & \textbf{\textbf{0.00}} & Nothing & \textbf{100} & \textbf{100} & nothing & \textbf{\textbf{0.00}} & \textbf{8.64} \\
			 & & &  \textsc{arc} & 0 & 83 & 0.2 & NOTHING & 83 & 83 & Nothing & \textbf{0.01} & 252.44 \\ 
			 && & a\textsc{arc} & 0 & 89 & \textbf{0.12} & Nothing & \textbf{89} & \textbf{89} & Nothing & \textbf{0.01} & 180.73 \\
			& & & \textsc{R\_cusum} &   0 & 100 & \textbf{\textbf{0.00}} & Nothing & \textbf{100} & \textbf{100} & Nothing & \textbf{\textbf{0.00}} & \textbf{28.78} \\

			\vspace{1em}
			& & & \textsc{R\_UStat} & 0 & 70 & 0.43 & Nothing & 70 & 0 & Nothing & \textbf{0.00} &1032.5\\

	0.2 & 1 & 1.2  & \textsc{pelt} & 100 & 0& 1.00 & nothing & 0 &0 & nothing & 1.00 & 5000 \\
			& && \textsc{Biweight}(2) & 3 & 0 & 8.01 &nothing& 1 & 1 &nothing& 0.43 & 2088.54 \\ 
			& & & \textsc{Biweight}(5)  & 99 & 1 & 0.99 & Nothing & 0 & 0 & nothing & 1.00 &4974.61 \\
			 & & &  \textsc{arc} & 9 & 47 & 0.69 & NOTHING & 42 & \textbf{21} & Nothing & 0.13 & \textbf{1103.67} \\ 
			 && & a\textsc{arc} & 15 & 41 & 1.03 & Nothing & 33 & 17 & Nothing & 0.24 & 1503.22 \\
			& & & \textsc{R\_cusum} &   6 & 94 & \textbf{0.06} & Nothing & \textbf{92} & \textbf{81} & Nothing & \textbf{0.01} & \textbf{217.1} \\
			& & & \textsc{R\_UStat} & 18 & 68 & \textbf{0.38} & Nothing & \textbf{68} & 3 & Nothing & \textbf{0.00} & 3597.5\\

			0.2&1&1.6 & \textsc{pelt} & 94 & 0& 1.00 & nothing & 0 &0 & nothing & 1.00 & 4773.49 \\
			& && \textsc{Biweight}(2) & 0 & 0 & 27.61 &nothing& 0 & 0 &nothing& 0.47 & 2325.96 \\ 
			& & & \textsc{Biweight}(5)  & 0 & 98 & 0.04 & Nothing & \textbf{98} & \textbf{98} & nothing & \textbf{\textbf{0.00}} & \textbf{30.63} \\
			 & & &  \textsc{arc} & 3 & 97 & \textbf{0.03} & NOTHING & 97 & 93 & Nothing & \textbf{0.01} & 186.51 \\ 
			 && & a\textsc{arc} & 1 & 94 & 0.08 & Nothing & 94 &90 & Nothing & \textbf{0.01} & 164.7 \\
			& & & \textsc{R\_cusum} &   0 & 100 & \textbf{\textbf{0.00}} & Nothing & \textbf{100} & \textbf{100} & Nothing & \textbf{\textbf{0.00}} &\textbf{16.18}\\
			& & & \textsc{R\_UStat} & 0& 84 & 0.20 & Nothing & 84 & 4 & Nothing & \textbf{0.00} & 3152.5\\

			0.2 & 2 & 1.2 & \textsc{pelt} & 100 & 0& 3.00 & nothing & 0 &0 & nothing & 1.00 & 5000 \\
			& && \textsc{Biweight}(2) & 7 & 4 & 6.32 &nothing& 0 & 0 &nothing& 0.24 & 1653.19\\ 
			& & & \textsc{Biweight}(5)  & 100 & 0 & 3.00 & Nothing & 0 & 0 & nothing & 1.00 & 5000 \\
			 & & &  \textsc{arc} & 0 & 29 & \textbf{0.88} & NOTHING & \textbf{43} & \textbf{2} & Nothing & \textbf{0.09} & \textbf{722.8} \\ 
			 && & a\textsc{arc} & 3 & 29 & \textbf{1.09} & Nothing & \textbf{23} & 1 & Nothing & \textbf{0.12} & \textbf{1011.52} \\
			& & & \textsc{R\_cusum} &   52 & 15 & 2.04 & Nothing & 15 & \textbf{12} & Nothing & 0.75 & 5000\\
			& & & \textsc{R\_UStat} & 56 & 23 & 2.03 & Nothing & 18 & 0 & Nothing & 1.00 & 3770\\

			0.2 & 2 & 1.6 & \textsc{pelt} & 94 & 0& 2.89 & nothing & 0 &0 & nothing & 1.00 & 4821.4 \\
			& && \textsc{Biweight}(2) & 0 & 0 & 28.27 &nothing& 1 &0 &nothing& 0.23 & 1106.68 \\ 
			& & & \textsc{Biweight}(5)  & 0 & 98 & 0.06 & Nothing & \textbf{99} & \textbf{98} & nothing & \textbf{\textbf{0.00}} &\textbf{20.07} \\
			 & & &  \textsc{arc} & 0 & 98 & \textbf{0.02} & NOTHING & 98 & 92 & Nothing & 0.01 & 90.98 \\ 
			 && & a\textsc{arc} & 0 & 95 & 0.05 & Nothing & 97 & 79 & Nothing & 0.01 & 109.88 \\
			& & & \textsc{R\_cusum} &  0& 100 & \textbf{\textbf{0.00}} & Nothing & \textbf{100} & \textbf{99}
			& Nothing & \textbf{\textbf{0.00}} & \textbf{31.54}\\
			& & & \textsc{R\_UStat} & 2 & 71 & 0.46 & Nothing & 69 & 0 & Nothing & 0.03 & 2427.5\\

			 \bottomrule
\end{longtable}

\subsubsection{Sensitivity of the choice of $h$}\label{sensitivity}
In the simulations above, we use a fixed $h$ for different choices of $\Delta$, which serves the purpose of testing the sensitivity of $h$.  To be complete, we also directly test the sensitivity of our methods with respect to the choice $h$ on three settings that are considered in the adversarial setting (i) and (ii) from the previous section. We consider five choices of window width $2h$ from the set $\{10\log(n), 20\log(n), 30\log(n), 40\log(n), 60\log(n)\} = \{85, 170, 255, 340, 511\}$. The results in the tables below are obtained by averaging over 100 repetitions and the numbers in the brackets indicate standard errors. The results suggest that a range of choices of the window width can achieve the best performance, whereas if $h$ is chosen to be too small or large relative to $L$, then the assumptions in our Theorem \ref{scp2} would be violated and lead to poor performance of the algorithms.

\begin{enumerate}
    \item Scenario (ii) with $\kappa = 1, \varepsilon = 0.1$ and $\Delta = 2$, which corresponds to $K = 3$ and $L = 1250$.
    \begin{table}[ht]
\centering
\begin{tabular}{@{}llll@{}}
\toprule
Algorithm & Choice of $2h$         & scaled Hausdorff distance & Number of change points \\ \midrule
\multirow{5}{*}{\begin{tabular}[c]{@{}l@{}}ARC \end{tabular}} &85 & 0.11 (0.10)                    & 3.42 (1.10)                   \\ 
&170& 0.05 (0.05)                    & 3.59 (0.89)                   \\ 
&255   & 0.01 (0.02)                    & 3.05 (0.22)                  \\ 
&340       & 0.01 (0.01)                   & 3.00(0.00)                   \\ 
&511      & 0.01 (0.00)                   & 3.00 (0.00)                  \\ \midrule
\multirow{5}{*}{\begin{tabular}[c]{@{}l@{}}aARC\end{tabular}}& 85 & 0.12 (0.09)                    & 4.79 (2.20)                   \\ 
&170& 0.06 (0.08)                    & 3.51 (1.19)                   \\ 
&255   & 0.02 (0.05)                    & 3.03 (0.41)                  \\ 
&340       & 0.02 (0.06)                   & 2.95 (0.26)                   \\ 
&511      & 0.03 (0.06)                   & 2.94 (0.24)                  \\ \bottomrule
\end{tabular}
\end{table}

\item Scenario (ii) with $\kappa = 1, \varepsilon = 0.1$ and $\Delta = 3$, which corresponds to $K = 5$ and $L = 833$.
     \begin{table}[ht]
\centering
\begin{tabular}{@{}llll@{}}
\toprule
Algorithm & Choice of $2h$         & scaled Hausdorff distance & Number of change points \\ \midrule
\multirow{5}{*}{\begin{tabular}[c]{@{}l@{}}ARC \end{tabular}} &85 & 0.09 (0.08)                    & 5.17 (1.21)                   \\
&170& 0.02 (0.03)                    & 5.32 (0.57)                   \\ 
&255   & 0.01 (0.00)                    & 5.00 (0.00)                  \\ 
&340       & 0.01 (0.02)                   & 4.99(0.10)                   \\ 
&511      & 0.17 (0.01)                   & 3.01 (0.10)                  \\ \midrule
\multirow{5}{*}{\begin{tabular}[c]{@{}l@{}}aARC\end{tabular}}& 85 & 0.09 (0.07)                    & 5.98 (1.96)                   \\ 
&170& 0.09 (0.04)                    & 5.27 (0.90)                   \\ 
&255   & 0.03 (0.06)                    & 4.83 (0.49)                  \\ 
&340       & 0.02 (0.03)                   & 4.95 (0.21)                   \\ 
&511      & $\infty$ (NaN)                   & 2.88 (0.41)                  \\ \bottomrule
\end{tabular}
\end{table}

\item Scenario (i) with $\sigma = 1, \varepsilon = 0.1$ and $ \Delta = 5$, which corresponds to $K = 0$ and $L = 5000$.
     \begin{table}[ht]
\centering
\begin{tabular}{@{}llll@{}}
\toprule
Algorithm & Choice of $2h$         & Number of change points \\ \midrule
\multirow{5}{*}{\begin{tabular}[c]{@{}l@{}}ARC \end{tabular}}  &85                    & 6.72 (3.34)                  \\
&170                  & 2.46 (1.92)                   \\ 
&255                     & 0.27 (0.65)                  \\ 
&340                       & 0.03(0.22)                   \\ 
&511                         & 0.00 (0.00)                  \\ \midrule
\multirow{5}{*}{\begin{tabular}[c]{@{}l@{}}aARC\end{tabular}} & 85                 & 7.08 (4.29)                   \\ 
&170                   & 2.15 (3.53)                   \\ 
&255                      & 1.10 (2.11)                  \\ 
&340                          & 0.78 (1.74)                   \\ 
&511                       & 0.18 (0.76)                  \\ 
\bottomrule
\end{tabular}
\end{table}
\end{enumerate}

\subsubsection{Less adversarial attack}\label{lessadv}
In below, we further consider two specific attacks which do not use any knowledge about the true change points and therefore less adversarial in nature. 

\begin{enumerate}
    \item We consider the contamination distributions to be the Normal distributions with means $2\sin(10\pi t/T)$ for $t = 1, ..., T$ and $\epsilon = 0.2$, where the sample size T = 3000. Three true change points are equally spaced and located at $750, 1500$ and $ 2250$. We fixed the signal-to-noise ratio to be $1.2$ with $\kappa = 1.2$ and $\sigma = 1$. Note that the relatively high frequency of the sin function creates the effect of spurious change points on a segment without true change points. We obtain the following results over 100 repetition and the numbers in the brackets indicate standard errors. 
    
    \begin{table}[ht]
\centering
\begin{tabular}{@{}lll@{}}
\toprule
            & scaled Hausdorff distance & Number of change points \\ \midrule
Biweight(2) & 0.16 (0.03)                    & 9.37 (2.40)                   \\ 
Biweight(5) & 0.03 (0.05)                    & 3.24 (0.52)                   \\ 
R\_cusum    & 0.01 (0.02)                    & 3.02 (0.14)                  \\ 
ARC         & 0.05 (0.07)                   & 2.90 (0.30)                   \\
aARC        & 0.06 (0.09)                   & 2.88 (0.46)                  \\ 
PELT        & $\infty$ (NaN)               & 0.00 (0.00)                   \\ 
\bottomrule
\end{tabular}
\end{table}
    
    \item  We consider the contamination distributions to be Cauchy distribution with scale parameter $10$. The experiment set-up is the same as case 1 above. This heavy-tailed contamination has been considered by \cite{fearnhead2019} and the Biweight algorithm is designed and proven to be effective in this setting. Therefore, it is not surprised that they outperform other methods in the results obtained below. 
        \begin{table}[ht]
\centering
\begin{tabular}{@{}lll@{}}
\toprule
            & scaled Hausdorff distance & Number of change points \\ \midrule
Biweight(2) & 0.00 (0.00)                    & 3.08 (0.27)                   \\ 
Biweight(5) & 0.00 (0.00)                    & 3.00 (0.00)                   \\ 
R\_cusum    & 0.01 (0.00)                    & 3.03 (0.14)                  \\ 
ARC         & 0.01 (0.02)                   & 3.02 (0.14)                   \\ 
aARC        & 0.01 (0.01)                   & 3.00 (0.00)                  \\ 
PELT        & 0.24 (0.01)               & 81.74 (12.45)                   \\ 
\bottomrule
\end{tabular}
\end{table}
\end{enumerate}

Based on the results above, we note that the comparison is unfair for the PELT algorithm as it is not a robust algorithm. The R\_cusum algorithm perform competitively but this combination of wild binary segmentation and robust testing procedure has not been studied before (theoretically or empirically). The performances of our algorithms (ARC and aARC) under these less adversarial scenarios are also not bad.

\subsection{Details of Section \ref{realdata}}\label{realdatasupp}

Throughout the real data experiments, we consider choices of $h$ that are smaller than the one used in simulation due suspected short minimal segment lengths and adapt $\lambda = \max\{1.2\sigma\sqrt{5\log(n)h^{-1}}, 8\sigma\varepsilon\}$ to account for the inflated estimation error caused by using a smaller $h$, where $n$ always refers the sample size of the data. Also, we observed that it is not necessary to search for $4h$ local maximisers, which often leads to under estimation of the number of change points. Instead, we search for $2h$ local maximisers in the real data experiments. When implementing a\textsc{arc}, we specify the part of data that we use for selecting~$\varepsilon$ according to the strategy described in Section \ref{tuningsele}. When implementing \textsc{arc}, we manually input the value $\varepsilon$ and adapt $\lambda = \max\{1.2\sigma\sqrt{5\log(n)h^{-1}}, 8\sigma\sqrt{\varepsilon}\}$ to account for the larger asymptotic bias when assuming $F_i$'s are heavy-tailed. \textsc{R\_UStat} is not considered in real data analysis since its extension to multiple change point detection is still preliminary and in particular, no method for choosing the initial block size is available. 

\subsubsection{Well-log data}

\citet{burg2020evaluation} performed a systematic study on the performance of different change point detection algorithms on a selection of real world data sets, which includes the well-log data. Five human annotations were collected for each data set and the covering metric (see Definition \ref{covermetric}) is used to evaluate the distance between the estimated change points and human annotations.  

Given the suspected short minimal segment length in the data, we choose $2h = 10\log(n)$ and adjust $\lambda$ accordingly while using the 3500th to 4000th data point to select $\varepsilon$. Running the a\textsc{arc} algorithm 100 times, we achieve an average score of 0.807 under the covering metric, which is better than the result obtained by fine-tuning the \textsc{Biweight} algorithm as reported in \cite{burg2020evaluation}. We note that \textsc{R\_cusum} implemented with wild binary segmentation also performs competitively achieving a score of 0.849 with 500 intervals and a BIC-type threshold.

\subsubsection{PM2.5 index data}
 For both datasets, we choose $2h = 15\log(n)$ and adapt $\lambda$ to account for both the choice of $h$ and the difference between \textsc{arc} and a\textsc{arc}. We would like to emphasise that the variability of \textsc{arc} and a\textsc{arc} are mainly due the sample splitting step in the \textsc{RUME}, where we randomly split the $2h$ data points into two sets of size $h$. This step is useful both in terms of theoretical analysis and numerical performance. The variability of the algorithms can be stabilised when the change points possess a stronger signal strength and/or the sample size is large.  

\begin{table}[hbt]
\centering
\caption{Frequency of the number of change points detected by a\textsc{arc} and \textsc{arc} over 1000 repetitions}
\label{realtable_arc}
\vspace{1em}
\begin{tabular}{@{}lllll@{}}
\toprule
                          Data sets                                       & Method & $\hat{K} = 2$ & $\hat{K} = 3$ & $\hat{K} = 4$  \\ \midrule
\multirow{2}{*}{\begin{tabular}[c]{@{}l@{}}London PM2.5\end{tabular}} & \textsc{arc}    &  0  &  350 &  463 \\
                                                                      & a\textsc{arc}   &  0  &  790 & 206 \\[3pt]
Beijing PM2.5 (corrupted)                                             & a\textsc{arc}   &  753  & 126  &  33 \\ \bottomrule
\end{tabular}
\end{table}

For the London PM2.5 data, we run a\textsc{arc} and \textsc{arc} 1000 times and use the 1000th to 1500th data points to select $\varepsilon$ in a\textsc{arc}. The result is shown in Table \ref{realtable_arc}. Out of the 1000 repetitions, the three change points corresponding to the a\textsc{arc} result denoted in Figure \ref{realdataplot} are detected 790 times. One additional change point is detected in the first quarter of the data 206 times. While for \textsc{arc} with input $\varepsilon = 0.01$, four change points, which corresponds to the \textsc{arc} result denoted in Figure \ref{realdataplot}, are detected 463 times. \textsc{Biweight}(2) and \textsc{Biweight}(3) seem to detect spurious change points caused by the large variability in the data set, as shown in Figure \ref{london_supp}.

For the Beijing PM2.5 data set, we run a\textsc{arc} 1000 times on the corrupted data set and use 500th to 1000th data points to select $\varepsilon$. The result is shown in Table \ref{realtable_arc}. The original two change points that were detected on the original data set (without contamination) are detected 753 times on the corrupted the data set and only one of them is detected 86 times. There are 126 times when the spurious change point is detected additionally. 

\begin{figure}[hbt]
    \centering
    \includegraphics[width = \linewidth,height = 1.5in]{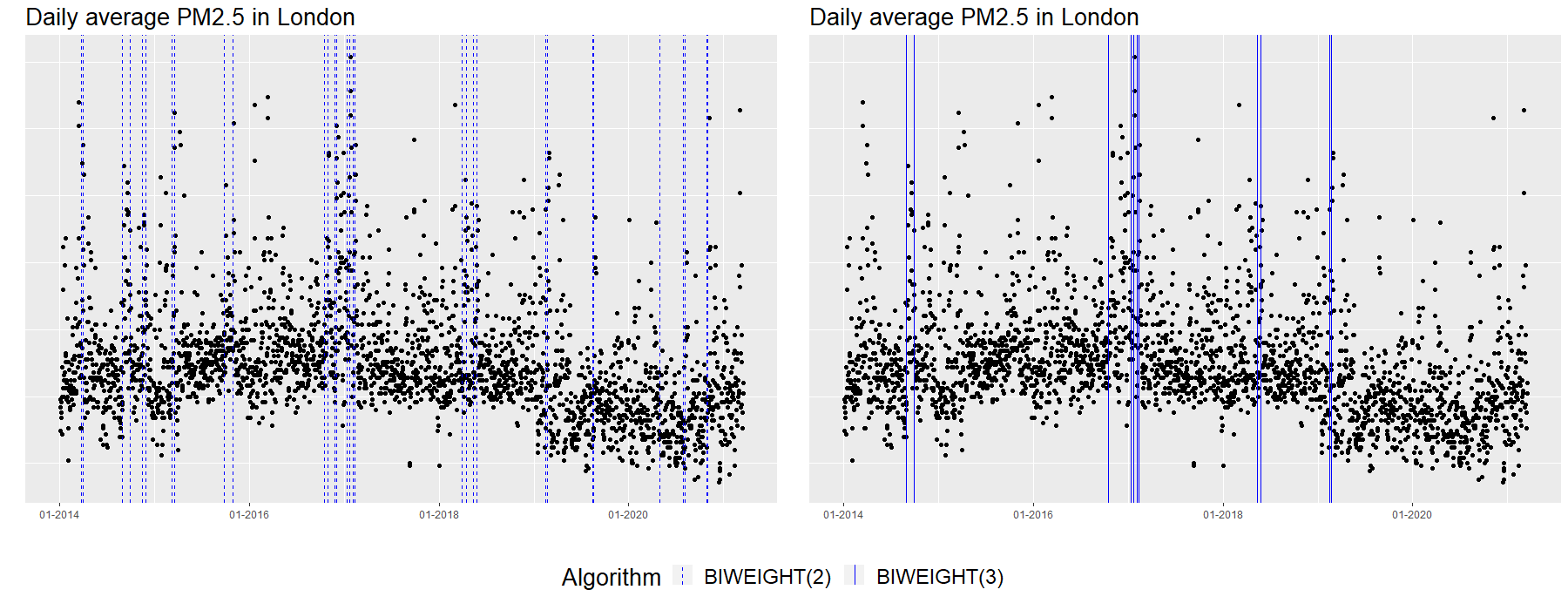}
    \caption{\textsc{Biweight} method on London PM 2.5 data with different penalty values}
    \label{london_supp}
\end{figure}

An inspection of the plot of the diagnostic statistics on the corrupted and clean data sets in Figure \ref{cusum} reveals that the original signal is restored in the presence of adversarial attack due to the local nature of the scanning method, while the fake signal created by the adversary will also exceed the threshold occasionally. We also consider the \textsc{Biweight} method with three different penalty values on the Beijing PM2.5 data. The result is shown in Figure \ref{biweightsupp}. Although different penalty values lead to different segmentation of the original data set, they all detect the spurious change point created by the adversarial noise after the data is contaminated. We use the result of \textsc{Biweight}(3) for illustration in Figure \ref{teaser} in Section \ref{sec-intro}, but we remove the last detected change point for clarity as it seems to be a spurious point near the endpoint of the data set.

\begin{figure}
\centering

\begin{subfigure}{0.45\textwidth}
  \centering
  \includegraphics[width=\linewidth]{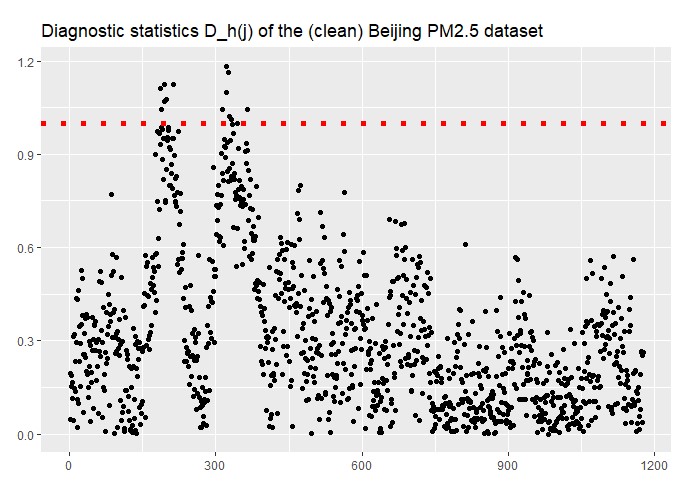}
\end{subfigure}
\begin{subfigure}{0.45\textwidth}
  \centering
  \includegraphics[width=\linewidth]{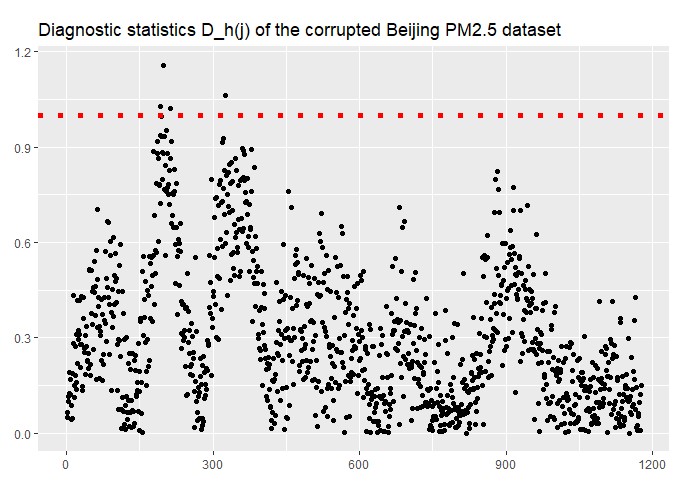}
\end{subfigure}
\caption{Plot of the diagnostic statistics $D_h(j)$ on the clean and corrupted Beijing PM2.5 index data}

\label{cusum}
\end{figure}

\begin{figure}[hbt]
    \centering
    \includegraphics[width=\linewidth]{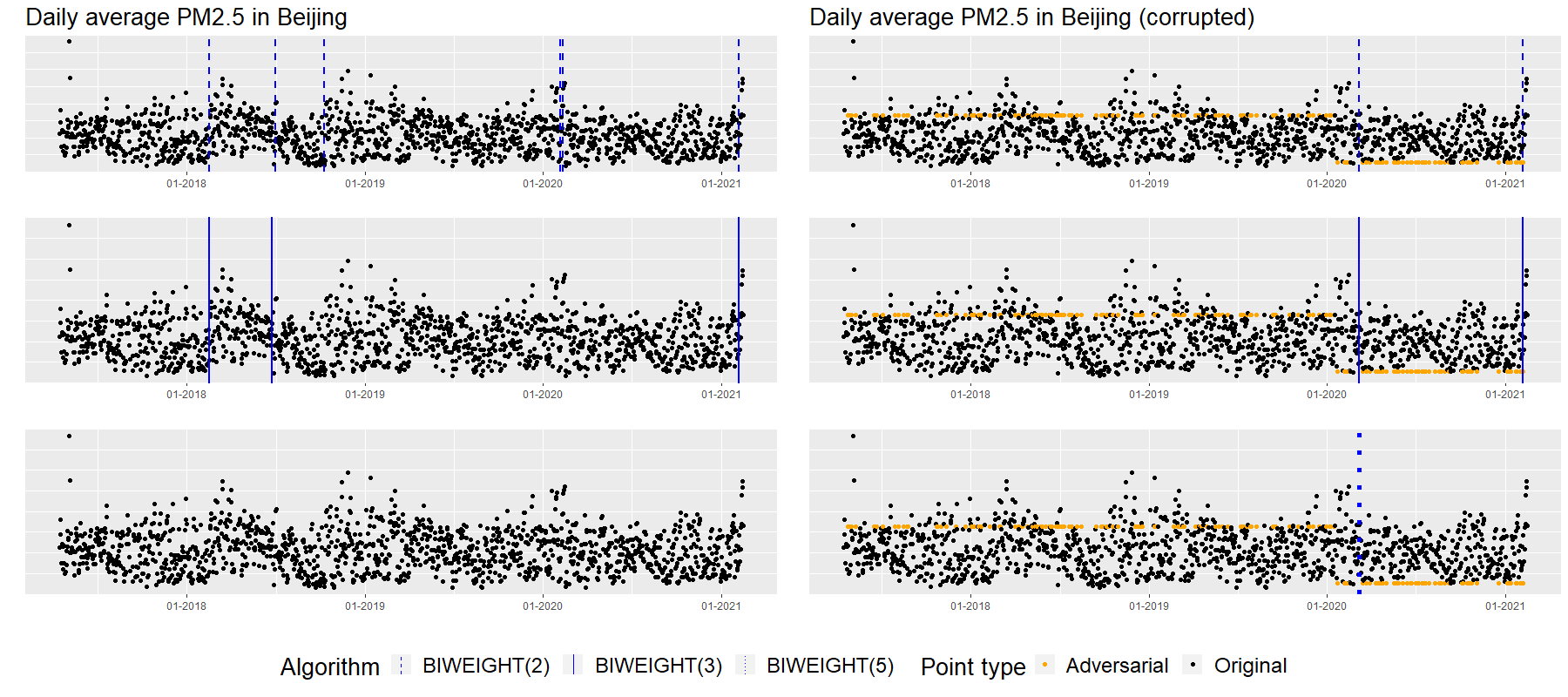}
    \caption{\textsc{Biweight} method on Beijing PM2.5 data with different penalty values}
    \label{biweightsupp}
\end{figure}

\end{document}